\documentclass[11pt]{article}
\usepackage{fullpage}
\usepackage{amsthm,amsfonts}
\usepackage{amsmath}
\usepackage{amsthm}
\usepackage{listings} 
\usepackage{xcolor}   
\usepackage{caption}
\usepackage{tikz}
\usetikzlibrary{arrows.meta}

\usepackage{algorithm}
\usepackage{algpseudocode}
\usepackage{amssymb}
\usepackage{color}
\usepackage{mathrsfs}
\usepackage{enumitem}
\usepackage{bm}
\usepackage{multirow}
\usepackage{booktabs}
\usepackage{makecell}
\usepackage{graphicx}
\usepackage{comment}
\usepackage{cases}
\usepackage{appendix}
\usepackage{tikz}
\usetikzlibrary{arrows,shapes}
\usepackage{marginnote}
\usepackage{tcolorbox}
\usepackage[qm]{qcircuit}
\usepackage{braket}
\usepackage{url}

\colorlet{color1}{blue}
\colorlet{color2}{red!50!yellow}

\usepackage{hyperref}[6.83]
\hypersetup{
  colorlinks=true,
  frenchlinks=false,
  pdfborder={0 0 0},
  naturalnames=false,
  hypertexnames=false,
  breaklinks,
  allcolors = color1,
  urlcolor = color2,
}

\usepackage{cleveref}

\numberwithin{equation}{section}

\usepackage{graphicx}
\usepackage{anyfontsize}
\usepackage{amsmath}
\usepackage{xcolor}
\usepackage{subcaption}
\usepackage{caption}
\usepackage{graphicx}
\usepackage{booktabs}
\usepackage{pifont}
\usepackage{amssymb}
\usepackage{algorithm}
\usepackage{algpseudocode}
\usepackage{amsfonts}
\usepackage{mathrsfs}

\newtheorem{theorem}{Theorem}
\newtheorem{lemma}[theorem]{Lemma}
\newtheorem{proposition}[theorem]{Proposition}

\newtheorem{definition}[theorem]{Definition}
\newtheorem{result}[theorem]{Result}

\newcommand{\REV}[1]{{#1}}

\allowdisplaybreaks
\title{Contour-integral based quantum eigenvalue transformation: analysis and applications} 

\author{Shan Jiang$^{1}$, \quad Dong An$^{2}$ \\ 
\footnotesize $^{1}$ School of Mathematical Sciences, Peking University, Beijing, China\\
\footnotesize $^{2}$ Beijing International Center for Mathematical Research, Peking University, Beijing, China\\
}

\date{ }

\begin{document}

\maketitle

\begin{abstract}
    Eigenvalue transformations appear ubiquitously in scientific computation, ranging from matrix polynomials to differential equations, and are beyond the reach of the quantum singular value transformation framework. In this work, we study the efficiency of quantum algorithms based on contour integral representation for eigenvalue transformations from both theoretical and practical aspects. Theoretically, we establish a complete complexity analysis of the contour integral approach proposed in [Takahira, Ohashi, Sogabe, and Usuda. Quant. Inf. Comput., 22, 11\&12, 965--979 (2021)]. Moreover, we combine the contour integral approach and the sampling-based linear combination of unitaries to propose a quantum algorithm for estimating observables of eigenvalue transformations using only $3$ additional qubits. Practically, we design contour integral based quantum algorithms for \REV{linear ordinary differential equations, $\varphi$-functions, Fr\'echet derivatives of matrix functions, and matrix functions. In particular, we} show that the contour integral algorithm can outperform all the existing quantum algorithms in the case of solving asymptotically stable differential equations. 
\end{abstract}

\tableofcontents

\section{Introduction}

Computing matrix functions is a crucial task in scientific computation. 
On the one hand, many matrix functions, including matrix inverse, matrix polynomials, and matrix exponential, have wide applications in various disciplines. 
On the other hand, matrix functions usually serve as a subroutine in iterative methods for solving non-linear equations and optimization algorithms. 
As computing matrix functions is typically the most costly step in the algorithm, efficient matrix function algorithms have a significant impact on accelerating the computation. 
In practice, the matrix is usually of high dimension. 
Such high dimensionality occurs ubiquitously when we use the model with a large degree of freedom for accuracy. 
This poses significant computational challenges because classical algorithms usually scale at least linear in terms of the dimension. 

To overcome the issue of high dimensionality, recent works turn to explore the power of quantum computers, which are expected to perform computational tasks by means of quantum superposition and entanglement, a very different way from classical computers. 
During the past decades, there have been many efficient quantum algorithms for certain types of matrix functions, such as matrix inverse (i.e., solving linear systems of equations)~\cite{HarrowHassidimLloyd2009,ChildsKothariSomma2017,SubasiSommaOrsucci2019,AnLin2022,CostaAnSandersEtAl2022,Dalzell2024,Li2025} and matrix exponentials (i.e., solving linear ordinary differential equations)~\cite{Berry2014,BerryChildsOstranderWang2017,ChildsLiu2020,Krovi2022,FangLinTong2023,BerryCosta2022,JinLiuYu2022,AnLiuLin2023,AnChildsLin2023,ShangGuoAnZhao2025,LowSomma2025}. 
Under a strong and somewhat debatable assumption that the input matrix has been efficiently encoded as specific quantum data structures, these quantum algorithms can implement the desired matrix functions with cost at most poly-logarithmic in dimension, offering a potential exponential speedup over the classical algorithms. 

A remarkable quantum algorithm among others is the quantum singular value transformation (QSVT)~\cite{GilyenSuLowEtAl2019}. 
As its name suggests, QSVT provides a uniform framework to implement any matrix function which can be expressed as a \emph{singular value transformation}. 
QSVT is both time- and space-efficient -- its overall query complexity does not have explicit large dimension dependence, and it only requires one additional ancilla qubit. 
Due to its versatility and high efficiency, QSVT has been widely applied to solve Hamiltonian simulation problem, matrix powers, Gibbs state preparation, eigenstate preparation, to name a few, and becomes a framework comprising major quantum algorithms~\cite{MartynRossiTanEtAl2021}. 

While certain matrix functions of practical interest, such as matrix inverse, can be expressed as singular value transformations, many of them cannot. 
A prototypical example is the matrix polynomial, which is intrinsically an \emph{eigenvalue transformation} rather than a singular value transformation if the matrix is not normal, because eigenvalue decomposition of a non-normal matrix is different from its singular value decomposition. 
Another significant instance of eigenvalue transformations is the matrix exponential that arises in solving differential equations. 
These problems are beyond the reach of QSVT, so designing efficient quantum eigenvalue transformation algorithms becomes the key to further unlocking the potential of quantum algorithms for high-dimensional scientific computation. 

In this work, we study the efficiency of quantum algorithms based on contour integral representation for eigenvalue transformation problems from both theoretical and practical aspects. 
Theoretically, we establish a detailed and complete complexity analysis of the quantum contour integral approach proposed in~\cite{TakahiraOhashiSogabeEtAl2020,TakahiraOhashiSogabeEtAl2021}. 
Our complexity analysis can be readily applied to holomorphic functions over an arbitrary domain that encompasses all the eigenvalues of the matrix. 
Moreover, we combine the contour integral approach and the recently proposed sampling-based linear combination of unitaries (LCU) technique~\cite{Chakraborty2024,WangMcArdleBerta2024} to propose a qubit-efficient quantum eigenvalue transformation algorithm, which only requires $3$ additional ancilla qubits to estimate an observable of an eigenvalue transformation. 

Practically, we design contour integral based quantum algorithms for applications including \REV{linear ordinary differential equations, $\varphi$-functions, Fr\'echet derivatives of matrix functions, and matrix polynomials.} 
Remarkably, we show that quantum contour integral approach can solve linear differential equations with optimal precision dependence. 
Furthermore, when the differential equations satisfy the asymptotically strictly stable condition, we show that the contour integral approach can outperform all the existing quantum algorithms. 
\REV{Besides, we also demonstrate advantage of the contour integral approach in implementing the matrix Schur polynomial, and discuss its even broader applications in computing matrix sign function, Riesz projectors, among others.}

\subsection{Problem statement}

Let $\ket{\psi}$ be an $N$-dimensional input state and $A$ be an $N \times N$ matrix. 
Our aim is to design a quantum algorithm to approximate $f(A)\ket{\psi}$ for a complex function $f(z)$. Suppose that \( \alpha \geq \|A\| \), and \( U_A \) is an $(\alpha, a, 0)$-block-encoding of the matrix \( A \), which means \( U_A = \begin{bmatrix} A/\alpha & \cdot \\ \cdot & \cdot \end{bmatrix} \) \REV{block encodes $A$ exactly with} $a$ ancilla qubits. \REV{We assume an exact block-encoding only for notational simplicity. If \(U_A\) is instead an \((\alpha,a,\delta)\)-block-encoding, then \(\delta\) can usually be chosen sufficiently small, and the cost of achieving such a precision is typically only polylogarithmic in \(1/\delta\).} Assume there is an oracle $U_\psi$ such that $U_\psi \ket{0} = \ket{\psi}$. Using $U_A$ and $U_\psi$, we define the problem for matrix functions as follows. 

\begin{definition}[Quantum matrix function problem]
\label{definition:main_problem}
    Let $A \in \mathbb{C}^{N \times N}$ be an $N \times N$ matrix with block-encoding $U_A$ and let $\ket{\psi}$ be a quantum state representing an $N$-dimensional complex unit vector. For matrix function $f(A)$ and state $\ket{\psi}$, we define quantum state $|f\rangle$ as $|f\rangle=f(A)\ket{\psi} / \| f(A)\ket{\psi} \|$. Then, for some positive constant $\epsilon$, the problem is to output a quantum state $|\tilde{f}\rangle$ such that
$$
\left\||f\rangle-|\tilde{f}\rangle \right\| \leq \epsilon
$$
with a probability of at least $1 / 2$, where $0 \leq \epsilon \leq 1 / 2$.
\end{definition}

In some cases, we do not need to obtain the vector $\ket{f}$, but only its value under a certain observation $O$. 
That is, we would hope to get an estimation of $\bra{\psi}f^\dagger(A) Of(A)\ket{\psi}$. At this point, we define the following problem. 

\begin{definition}[Quantum matrix function estimation problem]
\label{definition:estimation_problem}
Let $A \in \mathbb{C}^{N \times N}$ be an $N \times N$ matrix with block-encoding $U_A$ and let $\ket{\psi}$ be a quantum state representing an $N$-dimensional complex vector. 
Let $O \in \mathbb{C}^{N \times N}$ be a Hermitian matrix. 
For some positive constant $\epsilon, \delta$, the problem is to output $\mu$ such that
$$
\left|\mu-\bra{\psi}f^\dagger(A) Of(A)\ket{\psi}\right| \leq \epsilon
$$
with probability at least $1-\delta$.
\end{definition}
\REV{Throughout this work, $O$ represents a quantum observable matrix assumed to be Hermitian, which can represent most practically relevant physical observables including Hamiltonians, Pauli operators, and orthogonal projectors. 
The Hermiticity assumption also ensures that the target expectation value is real and can be estimated by the Hadamard test used in the algorithm later. 
We note that this assumption does not restrict the estimation of a general non-Hermitian operator, because a non-Hermitian operator can be decomposed into its Hermitian and anti-Hermitian parts and the two expectation values can be estimated separately.}

\subsection{Main results}

\subsubsection{Complexity analysis of the quantum contour integral algorithm}

\REV{Based on Cauchy's integral formula
\begin{equation}
    f(A)=\frac{1}{2 \pi \mathrm{i}} \oint_{\Gamma} f(z)\left(z I_N-A\right)^{-1} \mathrm{~d} z,
    \label{function:f(A)}
\end{equation}
we can implement a \emph{quantum contour integral algorithm},} where $\Gamma$ is a contour in the complex plane encompassing the eigenvalues of $A$. 
\REV{In the quantum contour integral algorithm}, we first apply any quadrature rule to discretize the integral into a finite sum. 
Then, the quantum algorithm proposed in~\cite{TakahiraOhashiSogabeEtAl2021} first implements the matrix inverse $(zI_N-A)^{-1}$ by QSVT, then performs quantum LCU~\cite{ChildsWiebe2012} to compute the finite sum from the integral discretization to approximate the desired matrix function. 

Our first main result is the complexity estimate of this quantum algorithm. 

\begin{result}[Informal version of Theorem~\ref{theorem:general_complexity}]\label{res:complexity_general}
    The quantum contour integral algorithm can solve the problem defined by Definition \ref{definition:main_problem} using
        $\widetilde{\mathcal{O}}\left(\frac{Bl\gamma^2 \alpha }{\left\|f(A)\ket{\psi}\right\|}\mathrm{log}\left(\frac{1}{\epsilon}\right)  \right)$\footnote{\REV{Here the notation $\widetilde{\mathcal{O}}$ hides the nondominant poly-logarithmic terms, i.e., $\widetilde{\mathcal{O}}(c(x)) = \mathcal{O}(c(x) \operatorname{polylog}(c(x)))$ for a scaling function $c(x)$. }}
     \REV{queries to $U_A$}, and 
    $\mathcal{O}\left(\frac{Bl\gamma}{\left\|f(A) \ket{\psi}\right\|}\right)$
     \REV{query to $U_{\psi}$}, and 
    $\mathcal{O}\left(\mathrm{log}\frac{(B\gamma+L)l}{\|f(A)\ket{\psi}\|\epsilon}\right)$
     additional ancilla qubits. 
     Here $\alpha \geq \|A\|$, $B$ is the maximum value of $|f(z)|$ on $\Gamma$, $L$ is the maximum value of $|f'(z)|$ on $\Gamma$, $l$ is the length of $\Gamma$, and $\gamma \geq \|(zI-A)^{-1}\|$ for all $z \in \Gamma$. 
\end{result}

Result~\ref{res:complexity_general} applies to any holomorphic function $f(z)$ and any contour $\Gamma$ that includes the eigenvalues of $A$. 
Due to its generality, Result~\ref{res:complexity_general} unavoidably involves many technical parameters which make it hard to understand the efficiency of the contour integral approach intuitively. 
We will show more explicit complexity estimates in terms of physical parameters for specific applications later. 

\subsubsection{Qubit-efficient quantum eigenvalue transformation algorithm}

Result~\ref{res:complexity_general} shows that the quantum contour integral algorithm needs logarithmically many ancilla qubits, due to the usage of the quantum LCU subroutine. 
For the problem of observable estimation in Definition~\ref{definition:estimation_problem}, we propose a hybrid quantum-classical algorithm which significantly reduces the number of ancilla qubits to constant. 
Our algorithm combines the contour integral representation and the sampling-based LCU in~\cite{Chakraborty2024}, while a similar idea has been used to design qubit-efficient hybrid quantum-classical algorithms for different tasks~\cite{WanBertaCampbell2022,DongLinTong2022,AnLiuLin2023,WangMcArdleBerta2024,WatsonWatkins2024}. 
Specifically, as the matrix function can be approximated by the linear combination of matrix inverses, the observable $\bra{\psi}f^\dagger(A) Of(A)\ket{\psi}$ can also be approximated by the linear combination of observables that only involve $O$ and two matrix inverses in the form of $(zI-A)^{-1}$ with different values of $z$. 
Then, we can use classical computers to sample the values of $z$ with probability related to the coefficients in the linear combination, independently estimate the sampled observable by QSVT (for matrix inverse) and Hadamard test (for observable estimation of $O$) on quantum computers, and take the sum of the sampled values classically. 

The complexity of our algorithm is given in the following result, which shows that our algorithm only requires constantly many ancilla qubits, while having worse overall query complexity compared to the purely quantum contour integral algorithm. 

\begin{result}[Informal version of Theorem~\ref{theorem:Ancilla-efficient_complexity}]

The hybrid quantum-classical contour integral algorithm can solve the problem defined by Definition \ref{definition:estimation_problem} using $\widetilde{\mathcal{O}} \left( \frac{\|O\|^2 B^4 l^4 \gamma^5 \alpha }{\epsilon^2 } \right)$ \REV{queries to $U_A$}, $\widetilde{\mathcal{O}}\left( \frac{\|O\|^2 (Bl\gamma)^4}{\epsilon^2  } \right)$ \REV{queries to $U_{\psi}$}, and $3$ additional ancilla qubits. 
Here the parameters are defined the same as in Result~\ref{res:complexity_general}. 
\end{result}

\subsubsection{Applications}

To understand the complexity of the quantum contour integral algorithm in a more explicit way and find its potential advantage over other algorithms, we apply the quantum contour integral algorithm to \REV{four} cases: \REV{linear ordinary differential equations, $\varphi$-functions, Fr\'echet derivatives of matrix functions, and matrix polynomials}. 

\paragraph{Linear differential equations. }We consider the homogeneous linear ordinary differential equation 
\begin{equation}
    \frac{\mathrm{d}}{\mathrm{d}t}x(t)=A x(t), \quad x(0)=\ket{\psi}.
\end{equation}
Its solution is given by $x(T) = \mathrm{e}^{AT} \ket{\psi}$, so we can solve the linear differential equation by implementing the eigenvalue transformation $f(A) = \mathrm{e}^{AT}$. 
It has been proved that generic quantum algorithms scale exponentially in the evolution time $T$ when solving unstable equations~\cite{AnLiuWangEtAl2023}. 
Therefore, to construct efficient algorithm, existing works always assume certain stable conditions. 
In this work, we consider two types of stable conditions. 

The first case is that $A = SDS^{-1}$ is diagonalizable and all its eigenvalues have non-positive real parts. 
To implement the contour integral algorithm, part of the contour that includes all the eigenvalues would unavoidably lies in the right half complex plane, since some of the eigenvalues are possibly imaginary. 
Then, Result \ref{res:complexity_general} shows that the complexity of the contour integral approach is 
\begin{equation}
\widetilde{\mathcal{O}}\left(\frac{\kappa_S^2 T^2}{\|x(T)\|} \log\left( \frac{1}{\epsilon} \right) \right), 
\end{equation}
which achieves near-optimal precision dependence (i.e., almost linear in $\log(1/\epsilon)$) but becomes sub-optimal in time $T$. 

\REV{Later in~\cref{compare:ODE}, we compare this bound with existing quantum algorithms for homogeneous linear ODEs.
The contour integral approach improves the explicit precision dependence from $\log^2(1/\epsilon)$ to $\log(1/\epsilon)$ relative to the time-marching and block-preconditioning bounds listed there, although its quadratic dependence on $T$ is worse than the nearly linear dependence of Taylor-based and
block-preconditioning methods. 
Approximate LCHS~\cite{LowSomma2025} also achieves logarithmic precision dependence with nearly linear time dependence, but requires a stronger condition $A+A^\dagger\preceq0$, whereas our bound applies to diagonalizable matrices with $\mathrm{Re}(\lambda_A)\leq0$, including matrices that do not satisfy this dissipativity condition (see~\cref{app:dissi_cond} for a comparison of different stable conditions).}

The second case is that $A = SDS^{-1}$ is diagonalizable and all its eigenvalues have strictly negative real parts. 
As all the eigenvalues are in the left half plane, we can now choose the contour to be a half circle in the left half plane combined with part of the imaginary axis. 
Then, Result \ref{res:complexity_general} shows that the complexity of the contour integral approach is only 
\begin{equation*}
    \widetilde{\mathcal{O}}\left( \frac{\kappa_S^2}{\|x(T)\|} \log\left(\frac{1}{\epsilon}\right)\right), 
\end{equation*}
which achieves a fast-forwarded scaling which does not have explicit dependence on the evolution time $T$. 
Our result is the first one to eliminate such explicit time dependence under the assumptions on the eigenvalues and outperforms existing generic and fast-forwarded quantum algorithms for differential equations~\cite{JenningsLostaglioLowrieEtAl2024}. 
See~\cref{compare:ODE} for a detailed comparison. 

\paragraph{\REV{$\varphi$-functions}}
\REV{The $\varphi$-functions arise naturally in the solution of inhomogeneous linear differential equations and are fundamental building blocks of exponential integrators for stiff ordinary differential equations and spatially discretized partial differential equations~\cite{HochbruckOstermann2010}. 
We consider the $\varphi$-functions 
\begin{equation}
    \varphi_k(z)
    =
    \sum_{j=0}^{\infty}\frac{z^j}{(j+k)!}
    =
    \frac{1}{(k-1)!}\int_0^1 \mathrm{e}^{(1-\theta)z}\theta^{k-1}\mathrm d\theta,
\end{equation}
and let $\varphi_0(z)=\mathrm{e}^z$. 
Then, Result \ref{res:complexity_general} shows that the contour integral approach can implement the matrix $\varphi$-functions, and its complexity is 
\begin{equation}
    \widetilde{\mathcal{O}}\left(\frac{\kappa_S^2}{k!\|\varphi_k(A)|\psi\rangle\|}\mathrm{log}\left(\frac{1}{\epsilon}\right)\right). 
\end{equation}
}

\REV{As a more concrete application, for inhomogeneous ODE
\begin{equation}
    \frac{\mathrm{d}}{\mathrm{d}t}x(t)=Ax(t)+b, 
    \quad x(0)=0,
\end{equation}
$T\varphi_1(TA)b$ acts as the inhomogeneous term. When $A$ is invertible and well-conditioned like the fast-forwarding case, $T\varphi_1(TA)=A^{-1}(\mathrm{e}^{AT}-I)$ can be achieved by a combination of $\mathrm{e}^{AT}$ and $A^{-1}$. However, for other cases like $A$ is semistable, $A^{-1}$ may be ill-conditioned or undefined, then we need to implement $T\varphi_1(TA)$ straightly. Here we assume $A = SDS^{-1}$ is diagonalizable and all its eigenvalues have non-positive real parts, and choose the contour in the same way as in the first case of linear ODE. Then, Result \ref{res:complexity_general} shows that the complexity of the contour integral approach is also
\begin{equation}
\widetilde{\mathcal{O}}\left(\frac{\kappa_S^2 T^2}{\|x(T)\|} \log\left( \frac{1}{\epsilon} \right) \right). 
\end{equation}}

\paragraph{\REV{Fr\'echet derivatives}}
\REV{Fr\'echet derivatives quantify the first-order sensitivity of matrix functions to perturbations and are important tools in condition estimation, model reduction, and quantum optimal control~\cite{HighamRelton2014Condition,deFouquieresEtAl2011}. 
Consider the Fr\'echet derivative of $f$ at $A$ in the direction $E$
\begin{equation}
    L_f(A,E)
    =
    \left.
    \frac{\mathrm d}{\mathrm d s}f(A+sE)
    \right|_{s=0}.
\end{equation}
This problem can be reduced to the matrix-function problem by implementing function $f$ on
\begin{equation}
    \mathcal A
    =
    \begin{pmatrix}
        A & E \\
        0 & A
    \end{pmatrix}
\end{equation}
and then postselecting the off-diagonal block. For simplicity, we assume $A = S D S^{-1}$ is diagonalizable and let $\kappa_S = \|S\|\|S^{-1}\|$ denote the condition number of $S$.  
We can choose the contour to be a circle centered at the origin with radius larger than the spectral radius of $A$. Then, Result \ref{res:complexity_general} shows that the complexity of the contour integral approach is
$$\widetilde{\mathcal{O}}\left(\frac{B\kappa_S^2(1+\kappa_S\|E\|)^2}{\left\|L_f(A,E)\ket{\psi}\right\|}\mathrm{log}\left(\frac{1}{\epsilon}\right) \right).$$}

\paragraph{Matrix polynomials. }Let $f(z)$ be a polynomial of $z$, and we would like to implement the matrix polynomial $f(A)$, which is naturally an eigenvalue transformation. 
For simplicity, we assume $A = S D S^{-1}$ is diagonalizable and let $\kappa_S = \|S\|\|S^{-1}\|$ denote the condition number of $S$.  
We can choose the contour to be a circle centered at the origin with radius larger than the spectral radius of $A$. 
Then, Result \ref{res:complexity_general} shows that the complexity of the contour integral approach is 
\REV{\begin{equation}
    \widetilde{\mathcal{O}}\left( \frac{B\log^2\kappa_S}{\|f(A)\ket{\psi}\|} \log\left( \frac{1}{\epsilon} \right)\right). 
\end{equation}}
Here again $B$ is the maximum value of the polynomial $|f(z)|$ over the contour. 
Compared to the existing quantum eigenvalue transformation algorithm with block preconditioning improvement~\cite{LowSu2024quantumlinearalgorithmoptimal}, the contour integral algorithm removes the explicit polynomial degree dependence \REV{and exponentially improves the $\kappa_S$ dependence} (See~\cref{compare:Matrix_Polynomial} for a detailed comparison). 
\REV{This enables us to implement matrix functions with good analytic properties like Schur functions at a low cost, even though their degrees might be high.}

\subsection{Related works}

The idea of leveraging contour integral representation to implement eigenvalue transformations was proposed in~\cite{TakahiraOhashiSogabeEtAl2020} and further improved in~\cite{TakahiraOhashiSogabeEtAl2021}. 
Besides, the works~\cite{TongAnWiebeEtAl2021,FangLinTong2023} also apply contour integral approach to some specific cases including solving differential equations and Gibbs state preparation, and a recent work~\cite{WangLiuXueEtAl2025} constructs a novel matrix decomposition identity based on Cauchy's residue theorem and designs an algorithm for simulating non-unitary dynamics and non-Hermitian matrix polynomials. 
The quantum contour integral approach studied in our work shares the same idea as in~\cite{TakahiraOhashiSogabeEtAl2021}, but our work gives a more complete complexity analysis which includes the integral discretization error and expresses the scalings explicitly in terms of the properties of the function to be implemented. 
Furthermore, we propose a qubit-efficient version of the algorithm, discuss the applications and find a case where the contour integral approach becomes the state of the art, which extends the scope of~\cite{TakahiraOhashiSogabeEtAl2021}. 

Besides the contour integral approach, there are two other quantum eigenvalue transformation algorithms: quantum eigenvalue processing (QEP) and Laplace transform based linear combination of Hamiltonian simulations (Lap-LCHS). 

QEP was proposed in~\cite{LowSu2024}, and the idea is to construct matrix Chebyshev polynomials by implementing its generating function with quantum linear system algorithms. 
In a recent work~\cite{LowSu2024quantumlinearalgorithmoptimal}, its query complexity to initial state preparation has been improved to be optimal by an improved linear system algorithm and the block-preconditioning technique. 
Compared to QEP, the contour integral approach has worse state preparation cost. 
However, the total complexity of the contour integral approach in terms of precision could be better than QEP. 
Specifically, to implement a matrix function other than polynomials, QEP needs to first approximate the desired function by Chebyshev polynomials and then implement the corresponding polynomials. 
The total complexity of QEP is $\mathcal{O}(d\log(1/\epsilon))$, where $d$ is the degree of the polynomial and is likely to contain extra $\epsilon$ dependence. 
On the other hand, the contour integral approach scales $\mathcal{O}(\log(1/\epsilon))$, because it directly approximates the function by the linear combination of matrix inverses, and the only step that introduces computational overhead is to implement the matrix inverses, which scales linearly in $\log(1/\epsilon)$ if we use QSVT. 

Lap-LCHS was proposed in~\cite{AnChildsLinying2024}, and the idea is to use Laplace transform to represent a matrix function as a linear combination of matrix exponentials, then further expand it as a linear combination of Hamiltonian simulation problems. 
Then the matrix function can be implemented by Hamiltonian simulation algorithms and LCU. 
Lap-LCHS can also achieve optimal state preparation cost in general, and it allows a qubit-efficient implementation as well by using sampling-based LCU. 
However, Lap-LCHS requires the function to have an inverse Laplace transform which is further needed to be integrable. 
This restricts the applicable regime of Lap-LCHS and excludes useful cases such as matrix polynomials. 
As a comparison, the contour integral approach only poses holomorphic condition on the function and can be more versatile.

\subsection{Discussions}

In this work, we study quantum algorithms for implementing matrix eigenvalue transformations based on contour integral representation. 
We establish a general complexity analysis of the quantum contour integral approach, and introduce a qubit-efficient way to estimate the observable with respect to the transformed state. 
We apply our analysis to solving Hamiltonian simulation problem, implementing matrix polynomials, and solving linear differential equations. 
We find that the contour integral approach can outperform all existing algorithms for solving linear asymptotically strictly stable differential equations by eliminating the explicit time dependence. 

Future directions include exploring more applications of the contour integral approach and constructing other quantum eigenvalue transformation algorithms with improved efficiency. 
Another promising direction is to explore whether the contour integral representation could be useful for designing a quantum algorithm to implement multi-matrix functions, which will be discussed in our follow-up work soon.

\subsection{Organization}

The rest of this work is organized as follows. 
In Section \ref{sec:prelim}, we briefly introduce some preliminary results needed in our analysis. 
Then we discuss the quantum and hybrid quantum-classical algorithms in Section \ref{sec:algorithm}, and give their complexity analysis in Section \ref{sec:complexity}. 
We show the applications in Section \ref{sec:applications}.

\section{Preliminaries}\label{sec:prelim}

In this section, we first introduce Cauchy's integral formula for matrix functions and discuss its numerical discretization. 
Then we introduce some basic quantum algorithms we will use later. 

\subsection{Cauchy's integral and numerical discretization}\label{section:Approximation by Cauchy’s integral}

Let $\Gamma$ be a closed contour in the complex plane that encloses all the eigenvalues of matrix $A$, and
let $f$ be an holomorphic function inside $\Gamma$. 
Then, using Cauchy’s integral formula, matrix function $f(A)$ can be described as in Eq.~\eqref{function:f(A)}. 
This gives us a way to approximate the matrix function $f(A)$ using numerical discrete approximations of integrals.

Normally, Eq. \eqref{function:f(A)}  holds if and only if $\Gamma$ is a curve enclosing all the eigenvalues of $A$. 
Let $\rho(A)$ be the \REV{spectral radius} of $A$. Because of $\rho(A)\leq \|A\|$, a circle with radius $\|A\|$ can cover all eigenvalues. But if $\Gamma$ is a circle with radius $\beta>\|A\|$, $f(z)$ may performs well around the eigenvalues of $A$, but performs poorly on disk $\{z:\|z\|\leq \beta\}$. Therefore, we want to relax $\Gamma$ to an arbitrary curve that covers all eigenvalues of $A$, and no longer assume the information of $f(z)$ outside $\Gamma$. 

\REV{Assume that $\Gamma$ is a piecewise $C^{1,1}$ simple closed curve covering all eigenvalues of $A$. 
Let $l$ be the length of $\Gamma$, $z = z(t)$ be its parametrization with $|z'(t)| = 1$ in every smooth piece for $t\in[0, l]$, and denote $z'(t)=e^{i\theta(t)}$. 
Let the rate of change of the tangent vector be bounded with $|e^{\mathrm i\theta(t)}-e^{\mathrm i\theta(t')}|
\le K_{\Gamma} |t-t'|$, and $f(z)$ is holomorphic in $\Omega$ and $L$-Lipschitz continuous in $\overline{\Omega}$.} By change of variable, we have 
\begin{equation}
    f(A)=\frac{1}{2 \pi \mathrm{i}} \int_{\Gamma} f(z)\left(z I_N-A\right)^{-1} \mathrm{~d} z = \frac{1}{2 \pi \mathrm{i}} \int_{0}^l f(z(t))\left(z(t) I_N-A\right)^{-1} \mathrm{e}^{\mathrm{i} \theta(t)}\mathrm{~d} t =\int_0^{l} h(t) \mathrm{d} t,
    \label{function:f(A)_general}
\end{equation}
where $h(t) = \frac{1}{2 \pi \mathrm{i}} f(z(t))\left(z(t) I_N-A\right)^{-1} \mathrm{e}^{\mathrm{i} \theta(t)}$, and $ \mathrm{e}^{\mathrm{i} \theta(t)} = \frac{\mathrm{d}}{\mathrm{d}t}z(t)$. 

Next, we approximate $f(A)$ by its Riemann sum $f_M(A)$ defined as 
\begin{equation}
f_M(A) = \sum_{k=0}^{M-1}\frac{l}{2 \pi \mathrm{i} M} f(z_k)\left(z_k I_N-A\right)^{-1} \mathrm{e}^{\mathrm{i} \theta_k}, 
\label{function:f_M(A)_general}
\end{equation}
where $M$ is the number of nodes in the discretization, and $t_k = \frac{lk}{M}$, $z_k = z(t_k)$ and $\theta_k =\theta(t_k)$ for $k = 0,1,\cdots, M$. 
The error of $f_M(A)$ approaching $f(A)$ can be bounded by the following proposition.

\begin{proposition}\label{proposition:general_discrete_error}
\REV{Suppose $\Gamma$ is a piecewise $C^{1,1}$ simple closed curve covering all eigenvalues of $A_j$ and $K_{\Gamma}$ is an upper bound on the tangent variation rate of $\Gamma$.} $f(z)$ is holomorphic inside $\Gamma$ and $L$-Lipschitz continuous inside and on $\Gamma$. Let $B$ be the maximum value of $|f(z)|$ on $\Gamma$, and assume that $\left\|(z(t) I_N -A)^{-1}\right\| \leq \gamma$ for any $t \in [0, l]$. 
Then, for the matrix function $f(A)$ in Eq.~\eqref{function:f(A)_general} and the approximation $f_M(A)$ in Eq.~\eqref{function:f_M(A)_general}, we have 
\begin{equation}\label{general_discrete_error}
    \left\|f(A)-f_M(A)\right\| \leq \frac{(B\gamma^2+B\REV{K_{\Gamma}}\gamma+L\gamma)l^2}{8\pi M}.
\end{equation}
\end{proposition}

\begin{proof}
For any $t_1, t_2 \in [0, l]$, 
\begin{equation}
\left\|(z(t_1) I_N -A)^{-1} - (z(t_2) I_N -A)^{-1}\right\| = \left\|\frac{(z(t_1)-z(t_2))I}{(z(t_1) I_N -A)(z(t_2) I_N -A)}\right\|\leq \gamma^2|t_1-t_2|.
\end{equation}
Because $f(z)$ is $L$-Lipschitz continuous on $\Gamma$, we have
\begin{equation}
\begin{aligned}
    \left\|h(t_1)-h(t_2)\right\| \leq &\left\|\frac{1}{2 \pi \mathrm{i}} f(z(t_1))\mathrm{e}^{\mathrm{i} \theta(t_1)}\right\| \times \left\|(z(t_1) I_N-A)^{-1}- (z(t_2) I_N-A)^{-1}\right\|+ \\ 
    &\left\|\frac{1}{2 \pi \mathrm{i}} f(z(t_1))-\frac{1}{2 \pi \mathrm{i}} f(z(t_2))\right\| \times \left\|\mathrm{e}^{\mathrm{i} \theta(t_1)} (z(t_2) I_N-A)^{-1}\right\|+ \\ 
    &\left\|\mathrm{e}^{\mathrm{i} \theta(t_1)}-\mathrm{e}^{\mathrm{i} \theta(t_2)}\right\| \times \left\|\frac{1}{2 \pi \mathrm{i}} f(z(t_2)) (z(t_2) I_N-A)^{-1}\right\| \\ 
    \leq & \frac{B\gamma^2+BK_{\Gamma}\gamma+L\gamma}{2\pi}|t_1-t_2|.
\end{aligned}
\end{equation}
Then we split the original integral into integrals over several small segments.
Note that $$\sum_{k=0}^{M-1}\left(t_{k+1}-t_k\right)h(t_k)=\sum_{k=0}^{M-1}\left(t_{k+1}-t_k\right)\frac{h\left(t_k\right)+h\left(t_{k+1}\right)}{2}.$$
For every segment, we have
\begin{equation}
\begin{aligned}
    &\left\|\int_{t_k}^{t_{k+1}} h(t) \mathrm{d} t 
- \left(t_{k+1}-t_k\right) \frac{h\left(t_k\right)+h\left(t_{k+1}\right)}{2}\right\|\\
\leq &\left\|\int_{t_k}^{\frac{t_k+t_{k+1}}{2}} \left(h(t)-h(t_k)\right) \mathrm{d} t+\int^{t_k+1}_{\frac{t_k+t_{k+1}}{2}} (h(t)-h(t_{k+1})) \mathrm{d} t\right\| \\
\leq & \int_{t_k}^{\frac{t_k+t_{k+1}}{2}} \frac{(t-t_k)(B\gamma^2+BK_{\Gamma}\gamma+L\gamma)}{2\pi} \mathrm{d} t+\int^{t_k+1}_{\frac{t_k+t_{k+1}}{2}} \frac{(t_{k+1}-t)(B\gamma^2+BK_{\Gamma}\gamma+L\gamma)}{2\pi} \mathrm{d} t \\
= & \frac{B\gamma^2+BK_{\Gamma}\gamma+L\gamma}{2\pi} \cdot \frac{(t_{k+1}-t_k)^2}{4}= \frac{(B\gamma^2+BK_{\Gamma}\gamma+L\gamma)l^2}{8\pi M^2}.
\end{aligned}
\end{equation}
Then summarizing all of these integrals, we can get the result
\begin{equation}
    \left\|f(A)-f_M(A)\right\|\leq\sum_{k=0}^{M-1}\left\|\int_{t_k}^{t_{k+1}} h(t) \mathrm{d} t 
- \left(t_{k+1}-t_k\right) \frac{h\left(t_k\right)+h\left(t_{k+1}\right)}{2}\right\|   \leq \frac{(B\gamma^2+BK_{\Gamma}\gamma+L\gamma)l^2}{8\pi M}.
\end{equation}

\end{proof}

\subsection{LCU lemma}
The LCU lemma gives a way to construct a linear combination of unitaries $\sum_{i=0}^{K-1}c_i U_i$ \cite[7.3]{Lin_qasc}. In the usual setting, for any $i$, $c_i$ is a non-negative real number. But in our analysis, we need to consider the case where $c_i$ is a complex number. Assume that  $c_i=r_i e^{\mathrm{i} \theta_i}$ with $r_i>0, \theta_i \in[0,2 \pi)$, define $\sqrt{c_i}=\sqrt{r_i} e^{\mathrm{i} \theta_i / 2}$ and the prepare oracles $V$ and $\widetilde{V}$:
\begin{equation}
    V=\frac{1}{\sqrt{\|c\|_1}}\left(\begin{array}{cccc}
\sqrt{c_0} & * & \cdots & * \\
\vdots & \vdots & \ddots & \vdots \\
\sqrt{c_{K-1}} & * & \cdots & *
\end{array}\right),
\end{equation}
\begin{equation}\label{matrix:tildeV}
    \widetilde{V}=\frac{1}{\sqrt{\|c\|_1}}\left(\begin{array}{ccc}
\sqrt{c_0} & \cdots & \sqrt{c_{K-1}} \\
* & \cdots & * \\
\vdots & \ddots & \vdots \\
* & \cdots & *
\end{array}\right).
\end{equation}
Together with the select oracle
\begin{equation}
    U:=\sum_{i \in[K]}\ket{i}\bra{i}\otimes U_i, 
\end{equation}
we can construct $\sum_{i=0}^{K-1}c_i U_i$ as follows. 
\begin{lemma}[LCU lemma]\label{lemma:lcu}
Let $c=(c_0,c_1,\cdots,c_{K-1})$ be a $K$-dimensional complex vector, and define $W=\left(\widetilde{V} \otimes I_n\right) U\left(V \otimes I_n\right)$. Then for any input state $|\psi\rangle$,
$$
W\left|0^{\mathrm{log}K}\right\rangle|\psi\rangle=\frac{1}{\|c\|_1}\left|0^{\mathrm{log}K}\right\rangle \sum_{i=0}^{K-1}c_i U_i|\psi\rangle+|\tilde{\perp}\rangle,
$$
where $|\widetilde{\perp}\rangle$ is an unnormalized state satisfying
$$
\left(\left|0^{\mathrm{log}K}\right\rangle\left\langle 0^{\mathrm{log}K}\right| \otimes I_n\right)|\widetilde{\perp}\rangle=0.
$$
\end{lemma}

In other words, $W$ is a $ (\|c\|_1, \mathrm{log}K, 0)$-block-encoding of $\sum_{i=0}^{K-1}c_i U_i$.

\subsection{Single-ancilla LCU}
\label{Single-Ancilla LCU}
Traditional LCU requires many ancilla qubits and complex multi-qubit controlled unitary operations, which are a challenge for early- and intermediate-term quantum computers. Recent work \cite{Chakraborty2024} significantly reduces the number of ancilla qubits in the LCU to one and simplifies the complex multi-qubit control gates, but requires the circuit to be repeated more times. Let $S=\sum_{i=0}^{K-1}c_i U_i$ and $\ket{\psi}$ be the initial state. Unlike the traditional LCU where we can get the quantum state $S\ket{\psi} / \| S\ket{\psi} \|$, the single-ancilla LCU aims to estimate the expectation value of an observable $O$. 

In single-ancilla LCU, we need $c_i$ to be positive real number. However, for complex $c_i$ we can always let $c_i=r_i e^{\mathrm{i} \theta_i}$ with $r_i>0, \theta_i \in[0,2 \pi)$, then $S=\sum_{i=0}^{K-1}c_i U_i=\sum_{i=0}^{K-1}r_i(e^{\mathrm{i} \theta_i}U_i)$ and $e^{\mathrm{i} \theta_i}U_i$ can be easily implemented. Therefore, without loss of generality, we can assume that $c_i$ is positive real number for $i=0,1,\cdots,K-1$.

In order to get the expectation of the observed value, we can select the unitary operator $U_i$ by randomly sampling from the probability distribution $\frac{1}{\|c\|_1}\{c_0, c_1, \cdots, c_{K-1}\}$ and acting on them on the initial state and observe the state under $O$. 
The specific process is shown by Algorithm \ref{algorithm:random_LCU}, which gives $\mu$ to be an estimation of $\bra{\psi}S^\dagger O S\ket{\psi}$. The accuracy of this algorithm is proved in Theorem \ref{theorem:Ancilla-efficient_complexity_proof}.
\begin{algorithm}

\caption{Expectation-observable $\left(O,U_k, \ket{\psi}, T\right)$}

1. Prepare the state $\ket{\psi_1}=\ket{+}\otimes\ket{\psi}$.

2. Obtain i.i.d. samples $V_1, V_2$ from the distribution $\left\{U_k, \frac{c_k}{\|c\|_1}\right\}$.

3. For $\tilde{V}_1=|0\rangle\langle 0| \otimes I+|1\rangle\langle 1| \otimes V_1$ and $\tilde{V}_2=|0\rangle\langle 0| \otimes V_2+|1\rangle\langle 1| \otimes I$, measure $(X \otimes O)$ on the state
$$
\ket{\psi^\prime}=\tilde{V}_2 \tilde{V}_1 \ket{\psi_1}
$$

4. For the $j^{\text {th }}$ iteration, store into $\mu_j$, the outcome of the measurement in Step 3.

5. Repeat Steps 1 to 4 , a total of $T$ times.

6. Output
$$
\mu=\frac{\|c\|_1^2}{T} \sum_{j=1}^T \mu_j .
$$
\label{algorithm:random_LCU}
\end{algorithm}

\subsection{Quantum singular value transformation}

Quantum Singular Value Transformation (QSVT, \cite{GilyenSuLowEtAl2019}) is an effective method for calculating matrix singular value transformation. It has different forms for different parity conditions. We use it here to solve the matrix inversion, so only its odd function version is introduced. For a square matrix $A \in C^{N \times N}$, where for simplicity we assume $N = 2^n$ for some positive integer $n$, the singular value decomposition (SVD) of the normalized matrix $A$ can be written as
\begin{equation}
    A = W \Sigma V^{\dagger}.
\end{equation}
\begin{definition}[Generalized matrix function] \label{definition:SVT_function}
Let $f: \mathbb{R} \rightarrow \mathbb{R}$ be an \REV{odd} scalar function such that $f\left(\sigma_i\right)$ is defined for all $i=1,2, \ldots, N$. The generalized matrix function is defined as
$$
f^{\diamond}(A):=W f(\Sigma) V^{\dagger},
$$
where
$$
f(\Sigma)=\operatorname{diag}\left(f\left(\sigma_1\right), f\left(\sigma_2\right), \ldots, f\left(\sigma_N\right)\right).
$$
\end{definition}
We assume there is an $(\alpha, m; 0)$-block-encoding of $A$ denoted by $U_A$, so that the singular values of $A/\alpha$ are in $[0, 1]$. For an $n$ degree real coefficient odd polynomial $p$ satisfying $p(x)\in[-1,1]$ for any $x\in[-1,1]$, QSVT provides an efficient way to implement $p^\diamond(A/\alpha)$ on a quantum computer using a simple circuit, with one ancilla qubit and $n$ queries to $U_A$.
\REV{Note that here we only define the generalized matrix function for odd functions, because in this work QSVT is used only to approximate the inverse function $1/x$. We refer to~\cite{GilyenSuLowEtAl2019} for a more general definition of the singular value transformation and its implementation. }

\section{Quantum algorithms}\label{sec:algorithm}

In this section, we introduce the detailed steps of the quantum algorithm for approximating $f_M(A)$. We first introduce the traditional version with the lower query complexity, then an ancilla-efficient version that simplifies the quantum circuit and uses fewer ancilla qubits. 

\subsection{Complexity-minimum algorithm}\label{section:Complexity-minimum Algorithm}

Our task is to approximate $f_M(A)$ in Eq. \eqref{function:f_M(A)_general} given $\Gamma$, $M$ and $A$. We start with the initial state $|\psi\rangle$. 
Note that we have block-encoding \( U_A \) of the matrix \( A \). 
Then
\begin{equation}
    U_A|0^a\rangle|\psi\rangle = \frac{1}{\alpha}|0^a\rangle A|\psi\rangle +\left|\perp\right\rangle.
\end{equation}

Let $z_k=z(t_k)$ be the $k$-th discrete point selected on $\Gamma$, $U_s$ be a select oracle, $V$ be a controlled rotation satisfying
\begin{equation}
    V|k\rangle|0\rangle_V = \frac{1}{\sqrt{\alpha + |z_k|}}|k\rangle\left(\sqrt{z_k}|0\rangle+\mathrm{i}\sqrt{\alpha}|1\rangle\right),
    \label{operator:control_rotation}
\end{equation}
\begin{equation}
    U_s = |0\rangle \langle0|\otimes I_N + |1\rangle \langle1|\otimes U_A.
\end{equation}
Thus, applying Lemma \ref{lemma:lcu}, we have that
\begin{equation}
\label{select_oracle}
    \left(\widetilde{V} \otimes I_a\otimes I_N\right) \left(I_k\otimes U_s\right)\left(V \otimes I_a\otimes I_N\right)|k\rangle |0\rangle_V|0^a\rangle|\psi\rangle = \frac{|k\rangle|0\rangle_V|0^a\rangle(z_k I_N-A)|\psi\rangle}{\alpha + |z_k|}+\ket{\widetilde{\perp}}, 
\end{equation}
where $\widetilde{V}$ is a control rotation, and its action on the last qubit is constructed from Eq. \eqref{matrix:tildeV}.
Among them, $\ket{\widetilde{\perp}}$ refers to both the vertical component in block-encoding and LCU. Thus we have that Eq. \eqref{select_oracle} is an $(\alpha + |z_k|, a+1, 0)$-block-encoding of $z_k I_N-A$ controlled by $\ket{k}$. 

As $x^{-1}$ is an odd function on $[-1,1]$ and singular at $x = 0$, for a small positive number $\delta$, we can use an odd polynomial $p(x)$ to approximate $\frac{3\delta}{4x}$ on $[-1,-\delta] \cup[\delta, 1]$, and when the precision is sufficient, $|p(x)|$ is less than $1$ on $[-1,1]$. Then $p^\diamond(\frac{z_k I_N-A}{\alpha + |z_k|})^\dagger$ can be implemented by QSVT, where $p^\diamond(U)$ is defined by Definition \ref{definition:SVT_function}. The control circuit of Eq. \eqref{select_oracle} can be kept on the QSVT circuit like Fig.\ref{fig:circuit_Control_QSVT}.
\begin{figure}
    \centerline{
    \Qcircuit @R=1em @C=1em {
    \ket{k} \quad\quad\quad & 
    \qw & \qw & \qw & \qw & \qw & 
    \ctrl{2} & \qw & \qw & \qw & \ctrl{2} &
    \qw & \cdots & & \qw & \qw & \qw & \qw & \qw & \qw &\\
    \ket{0} \quad\quad\quad &
    \qw & \gate{H} &\targ & \gate{\mathrm{e}^{-\mathrm{i} \phi_0 Z}} & \targ &
    \qw & \targ & \gate{\mathrm{e}^{-\mathrm{i} \phi_{1} Z}} & \targ & \qw & 
    \qw & \cdots & & \qw  &\targ & \gate{\mathrm{e}^{-\mathrm{i} \phi_d Z}} & \targ & \gate{H} & \qw & \\ 
    \ket{0}_V\ket{0^a} \quad\quad\quad\quad & 
    \qw & \qw & \ctrlo{-1} & \qw & \ctrlo{-1} & 
    \multigate{1}{U_k} & \ctrlo{-1} & \qw & \ctrlo{-1} & \multigate{1}{U_k^\dagger} & 
    \qw & \cdots & & \qw & \ctrlo{-1} & \qw & \ctrlo{-1} & \qw & \qw &  \\
    \ket{\psi} \quad\quad\quad & 
    \qw & \qw & \qw & \qw & \qw & 
    \ghost{U_k} & \qw & \qw & \qw & \ghost{U_k^\dagger} & 
    \qw & \cdots & & \qw & \qw & \qw & \qw & \qw & \qw &  \\
    }
    }
    \caption{QSVT with control circuit. Due to the control operator, it can also act as a select oracle $\mathrm{SEL}$ in outer layer LCU.}
    \label{fig:circuit_Control_QSVT}
\end{figure} 

Next, let
\begin{equation}
    c = \left(c_0, \cdots, c_{M-1}\right) :=\left(\frac{2l}{3\pi \mathrm{i}\delta M(\alpha+|z_0|)} f\left(z_0\right)\mathrm{e}^{\mathbf{i} \theta_0}, \cdots, \frac{2l}{3\pi \mathrm{i}\delta M(\alpha+|z_0|)} f\left(z_{M-1}\right) \mathrm{e}^{\mathbf{i} \theta_{M-1}}\right).
\end{equation} 
Then
\begin{equation}
\begin{aligned}
    \sum_{k=0}^{M-1} c_kp^\diamond\left(\frac{z_k I_N-A}{\alpha + |z_k|}\right)^\dagger &\approx \sum_{k=0}^{M-1}\frac{2l}{3\pi \mathrm{i}\delta M(\alpha+|z_k|)} f\left(z_k\right)\mathrm{e}^{\mathrm{i} \theta_k} \frac{3\delta}{4}\left(\frac{z_k I_N-A}{\alpha + |z_k|}\right)^{-1}\\
    &=\sum_{k=0}^{M-1}\frac{l}{2\pi \mathrm{i}M}  f\left(z_k\right)\mathrm{e}^{\mathrm{i} \theta_k}\left(z_k I_N-A\right)^{-1}=f_M(A). 
\end{aligned}
\end{equation}
We can generate the unitary $C$ such that
\begin{equation}
\label{prepare_oracle}
    C|0\rangle=\frac{1}{\sqrt{\|c\|_1}} \sum_{j=0}^{M-1} \sqrt{c_j}|j\rangle.
\end{equation}
Then, we can define the select oracle
\begin{equation}
    \mathrm{SEL} = \sum_{k=0}^{M-1} |k\rangle \langle k|\otimes W_k,
\end{equation}
where $W_k$ is the block-encoding of $p^\diamond\left(\frac{z_k I_N-A}{\alpha + |z_k|}\right)^\dagger$ controlled by $\ket{k}$. This oracle is exactly the oracle defined in Eq. \eqref{select_oracle}.
Then by Lemma \ref{lemma:lcu}, the operator 
\begin{equation}\label{operator:total_circuit}
     (\widetilde{C} \otimes I )\mathrm{SEL}\left(C\otimes I\right)
\end{equation}
gives a block-encoding of $\sum_{k=0}^{M-1} c_kp^\diamond\left(\frac{z_k I_N-A}{\alpha + |z_k|}\right)^\dagger $ , which is an approximation of $f_M(A)$. Since the polynomial $p$ computed is the same for all $k$, the QSVT circuit is also the same for all $k$. Therefore, $\mathrm{SEL}$ is given by the circuit shown in Fig.\ref{fig:circuit_Control_QSVT}.

\subsection{Ancilla-efficient algorithm}\label{section:Ancilla-efficient Algorithm}

Section \ref{Single-Ancilla LCU} gives us a method to simplify the circuit and reduce the number of ancilla qubits. This method no longer gets an approximation of the target state $f(A)\ket{\psi} / \| f(A)\ket{\psi} \|$, but the observation value of the target state under the observable $O$. The question is defined by Definition \ref{definition:estimation_problem}.

In this method, we no longer use complex control gate $V$, but for every $k$, we implement
\begin{equation}
    V_k|0\rangle_V = \frac{1}{\sqrt{\alpha + |z_k|}}(\sqrt{z_k}|0\rangle+\mathrm{i}\sqrt{\alpha}|1\rangle),
\end{equation}
then follow the steps above. We can get $U_k$ to be an $(1,a+2)$-block-encoding of $p^\diamond\left(\frac{z_k I-A}{\alpha+|z_k|}\right)^\dagger$ for each $k$ respectively.

Let $\ket{\widetilde{\psi}}=\ket{0^{a+2}}\otimes\ket{\psi}$, then $U_k\ket{\widetilde{\psi}}=\ket{0^{a+2}}\otimes p^\diamond(\frac{z_k I-A}{\alpha+|z_k|})^\dagger\ket{\psi}+\ket{\perp}$. Let $\widetilde{O}=\ket{0^{a+2}}\bra{0^{a+2}}\otimes O$, then when $O$ is an Hermitian matrix, $\widetilde{O}$ is still Hermitian, and we have
\begin{equation}
   \begin{aligned}
\widetilde{O}U_k\ket{\widetilde{\psi}}&=\left(\ket{0^{a+2}}\bra{0^{a+2}}\otimes O\right)\left(\ket{0^{a+2}}\otimes p^\diamond\left(\frac{z_k I-A}{\alpha+|z_k|}\right)^\dagger\ket{\psi}+\ket{\perp}\right)\\
   &= \left(\ket{0^{a+2}}\otimes O p^\diamond\left(\frac{z_k I-A}{\alpha+|z_k|}\right)^\dagger\ket{\psi}\right),
\end{aligned}
\end{equation}

\begin{equation}
\begin{aligned}
\bra{\widetilde{\psi}}U_k^\dagger\widetilde{O}U_k\ket{\widetilde{\psi}}=&\left(\bra{0^{a+2}}\otimes\bra{\psi} p^\diamond\left(\frac{z_k I-A}{\alpha+|z_k|}\right)+\bra{\perp}\right)\left(\ket{0^{a+2}}\otimes O p^\diamond\left(\frac{z_k I-A}{\alpha+|z_k|}\right)^\dagger\ket{\psi}\right)\\
=&\bra{\psi} p^\diamond\left(\frac{z_k I-A}{\alpha+|z_k|}\right)O p^\diamond\left(\frac{z_k I-A}{\alpha+|z_k|}\right)^\dagger\ket{\psi}.
\end{aligned}
\end{equation}

Therefore, we can use the algorithm in Section \ref{Single-Ancilla LCU} with $\ket{\widetilde{\psi}}$ be the initial state and $\widetilde{O}$ be the observable. 
In the second step of Algorithm \ref{algorithm:random_LCU}, we can sample $V_1, V_2$ from distribution $\left\{U_k, \frac{c_k}{\|c\|_1}\right\}$. After the above steps in Algorithm \ref{algorithm:random_LCU}, we would get $\mu$ with
\begin{equation}
\begin{aligned}
    \mathbb{E}(\mu)=&\|c\|_1^2\left(\frac{1}{\|c\|_1}\right)^2\bra{\widetilde{\psi}}\left(\sum_k c_k U_k\right)^\dagger\widetilde{O}\left(\sum_kc_kU_k\right)\ket{\widetilde{\psi}}\\
    =&\bra{\psi}\left(\sum_k c_k p^\diamond\left(\frac{z_k I-A}{\alpha+|z_k|}\right)^\dagger\right)^\dagger O\left(\sum_kc_kp^\diamond\left(\frac{z_k I-A}{\alpha+|z_k|}\right)^\dagger\right)\ket{\psi}\\
    \approx&\bra{\psi}f^\dagger(A) Of(A)\ket{\psi}.
\end{aligned}
\end{equation}

\section{Complexity analysis}\label{sec:complexity}

In this section, we analyze the query complexity of the algorithm to $U_A$ and $U_\psi$, as well as the number of ancilla qubits used, under different conditions mentioned previously.

\subsection{Complexity-minimum algorithm}

In this part, we consider the algorithm given in Section \ref{section:Complexity-minimum Algorithm}.
Due to the relaxation of the holomorphic region, integral curve, and the norm condition about $A$ to the spectrum, we do not have exponential convergence of $\epsilon$ with respect to $M$. Fortunately, the size of $M$ does not affect the number of calls to  $U_A$ and $U_{\psi}$.

In the previous discussion, we mentioned that we can use $p(x)$ to fit $\frac{3\delta}{4x}$ on $[-1,-\delta] \cup[\delta, 1]$. Suppose that
\begin{equation}
    \left|p(x)-\frac{3 \delta}{4 x}\right| \leq \epsilon^{\prime}, \quad \forall x \in[-1,-\delta] \cup[\delta, 1]
\end{equation} 
and $|p(x)| \leq 1$ for all $x \in [-1,1]$. The existence of such an odd polynomial of degree $\mathcal{O}(\frac{1}{\delta}\mathrm{log}(\frac{1}{\epsilon^\prime}))$ is guaranteed by \cite[Corollary 69]{GilyenSuLowEtAl2019}. 

Assume that there is singular value decomposition $$\frac{z_k I_N-A}{\alpha + |z_k|}=V_k \frac{\Sigma_k}{\alpha+|z_k|} W_k^{\dagger},$$ then \cite[Theorem 2]{GilyenSuLowEtAl2019} enables us to implement $$\left(p^{\diamond}\left(\frac{z_k I_N-A}{\alpha + |z_k|}\right)\right)^{\dagger}=W_k p\left(\frac{\Sigma_k}{\alpha+|z_k|}\right) V_k^{\dagger}.$$ If all diagonal elements of $\frac{\Sigma_k}{\alpha+|z_k|}$, i.e., the singular values of $\frac{z_k I_N-A}{\alpha + |z_k|}$, are in the interval $[\delta , 1]$ for $k = 0,1,\cdots,M-1$, then we have 
\begin{equation}\label{QSVT_error}
    \left\|\left(p^{\diamond}\left(\frac{z_k I_N-A}{\alpha + |z_k|}\right)\right)^{\dagger}-(3 \delta / 4)\left(\frac{z_k I_N-A}{\alpha + |z_k|}\right)^{-1}\right\|=\left\|p\left(\frac{\Sigma_k}{\alpha+|z_k|}\right)-(3 \delta / 4)\left(\frac{\Sigma_k}{\alpha+|z_k|}\right)^{-1}\right\| \leq \epsilon^{\prime}.
\end{equation}

Assume $\Gamma$ is inside disk $|z| \leq R$, then $0\leq\alpha + |z_k| \leq \alpha+R$. Under the assumption that for any $t \in [0, l]$, $\left\|(z(t) I_N -A)^{-1}\right\| \leq \gamma$,  the singular values of $\frac{z_k I_N-A}{\alpha + |z_k|}$ are all larger than $\delta = \frac{1}{\gamma(\alpha+R)}$. Therefore, 
\begin{equation}
\|c\|_1 \leq \mathcal{O}\left(\frac{Bl}{\alpha\delta}\right)= \mathcal{O}(Bl\gamma(1+R/\alpha) ).
\end{equation}
Since we use $\sum_{k=0}^{M-1} c_kp^\diamond\left(\frac{z_k I_N-A}{\alpha + |z_k|}\right)^\dagger $ to approximate $f_M(A)$, we have
\begin{equation}
\label{equation:QSVT_error}
    \left\|\sum_{k=0}^{M-1} c_kp^\diamond\left(\frac{z_k I_N-A}{\alpha + |z_k|}\right)^\dagger -f_M(A)\right\| \leq\sum_{k=0}^{M-1}\|c_k\|\left\| p^\diamond\left(\frac{z_k I_N-A}{\alpha + |z_k|}\right)^\dagger -(3 \delta / 4)\left(\frac{z_k I_N-A}{\alpha + |z_k|}\right)^{-1}\right\|\leq \|c\|_1\epsilon^{\prime}.
\end{equation}
When the algorithm succeeds, we get the state $$\sum_{k=0}^{M-1} c_kp^\diamond\left(\frac{z_k I_N-A}{\alpha + |z_k|}\right)^\dagger\ket{\psi}/\left\|\sum_{k=0}^{M-1} c_kp^\diamond\left(\frac{z_k I_N-A}{\alpha + |z_k|}\right)^\dagger\ket{\psi}\right\|.$$

Also, choose a proper $M$ so that  
\begin{equation}\label{condition:M}
    \left\|f(A)\ket{\psi}-f_M(A)\ket{\psi}\right\| \leq \|c\|_1\epsilon^\prime,
\end{equation}
then combine these two equations together, we have
\begin{equation}
    \left\|f(A)\ket{\psi}-\sum_{k=0}^{M-1} c_kp^\diamond\left(\frac{z_k I_N-A}{\alpha + |z_k|}\right)\ket{\psi}\right\| \leq 2\|c\|_1\epsilon^\prime,
\end{equation}
and due to that $\|\frac{u}{\|u\|}-\frac{v}{\|v\|}\|\leq2\frac{\|u-v\|}{\|u\|}$, 
\begin{equation}
    \left\|\frac{f(A)\ket{\psi}}{\|f(A)\ket{\psi}\|}-\frac{\sum_{k=0}^{M-1} c_kp^\diamond\left(\frac{z_k I_N-A}{\alpha + |z_k|}\right)^\dagger\ket{\psi}}{\|\sum_{k=0}^{M-1} c_kp^\diamond\left(\frac{z_k I_N-A}{\alpha + |z_k|}\right)^\dagger\ket{\psi}\|}\right\| \leq \frac{4\|c\|_1\epsilon^\prime}{\|f(A)\ket{\psi}\|}.
\end{equation}
Therefore, we can let 
$$
\epsilon^\prime = \frac{\|f(A)\ket{\psi}\|\epsilon}{4\|c\|_1}.
$$
As $\|c\|_1 = \mathcal{O}(Bl\gamma(1+R/\alpha) )$, we have $$\frac{1}{\epsilon^\prime}= \mathcal{O}\left(\frac{Bl\gamma (1+R/\alpha) }{\|f(A)\ket{\psi}\|\epsilon}\right),$$ 
and the total number of calls to $U_A$ by $\mathrm{SEL}$ is 
\begin{equation}
    \mathcal{O} \left(\frac{1}{\delta}\mathrm{log}\frac{1}{\epsilon^\prime}\right) = \mathcal{O}\left(\gamma(\alpha+R)\mathrm{log}\left(\frac{Bl\gamma (1+R/\alpha) }{\|f(A)\ket{\psi}\|\epsilon}\right)\right).
\end{equation}

As the whole circuit is given by Eq. \eqref{operator:total_circuit}, we should consider $C,\widetilde{C}$ and $\mathrm{SEL}$ respectively when calculating complexity. For any given unit vector $\ket{\psi}$, the success probability of Eq. \eqref{operator:total_circuit} is 
$$
\left\|\sum_{k=0}^{M-1} c_kp^\diamond\left(\frac{z_k I_N-A}{\alpha + |z_k|}\right)^\dagger\ket{\psi}\right\|^2/\|c\|_1^2 .  
$$
Therefore, with amplitude amplification, 
$$
\mathcal{O}\left(\|c\|_1/\left\|\sum_{k=0}^{M-1} c_kp^\diamond\left(\frac{z_k I_N-A}{\alpha + |z_k|}\right)^\dagger\ket{\psi}\right\|\right)
$$
calls to Eq. \eqref{operator:total_circuit} can make the algorithm succeed with probability at least $1/2$. 

According to Proposition \ref{proposition:general_discrete_error}, in order to suit Eq. \eqref{condition:M} we need to choose $M$ so that 
\begin{equation}\label{eqn:M_choice}
    \frac{(B\gamma^2+B\REV{K_{\Gamma}}\gamma+L\gamma)l^2}{8\pi M}\leq\epsilon^\prime = \frac{\|f(A)\ket{\psi}\|\epsilon}{4\|c_1\|}, 
\end{equation}
which means
\begin{equation}
    \log M = \mathcal{O}\left(\log\left(\frac{(B\gamma+B\REV{K_{\Gamma}}+L)Bl\gamma^3}{\|f(A)\ket{\psi}\|\epsilon}\right)\right)=\mathcal{O}\left(\log\left(\frac{(B\gamma+\REV{BK_{\Gamma}}+L)l}{\|f(A)\ket{\psi}\|\epsilon}\right)\right).
\end{equation}
Since Eq. \eqref{operator:total_circuit} needs $\log M+a+2$ ancilla qubits, it's a $\left(\|c\|_1, \mathcal{O}\left(\log\frac{(B\gamma+L)l}{\|f(A)\ket{\psi}\|\epsilon}\right)+a+2, 2\|c\|_1\epsilon^\prime\right)$-block-encoding of $f(A)$.

Suppose that $\epsilon$ is small enough, then $\sum_{k=0}^{M-1} c_kp^\diamond\left(\frac{z_k I_N-A}{\alpha + |z_k|}\right)^\dagger\ket{\psi}$ and $f(A)\ket{\psi}$ have the same order. In other words, $$\|c\|_1/\left\|\sum_{k=0}^{M-1} c_kp^\diamond\left(\frac{z_k I_N-A}{\alpha + |z_k|}\right)^\dagger\ket{\psi}\right\|= \mathcal{O}\left(\frac{Bl\gamma (1+R/\alpha) }{\left\|f(A)\ket{\psi}\right\|}\right).$$ 
Summarizing these calculations, we can get the following result. 

\begin{theorem}
Suppose that $\Gamma$ is a piecewise $C^{1,1}$ simple closed curve covering
all eigenvalues of $A$ on the disk $|z| \leq R$, and $f(z)$ is an holomorphic function inside $\Gamma$ and $L$-Lipschitz continuous on $\Gamma$. Let $B$ be the maximum value of $|f(z)|$ on $\Gamma$. $U_A$ is an $(\alpha,a,0)$-block-encoding of A. Under the assumption that for any $z \in \Gamma$, $\left\|(z I_N -A)^{-1}\right\| \leq \gamma$, we can solve the problem defined by Definition \ref{definition:main_problem} using a quantum circuit with
$$\mathcal{O}\left(\frac{Bl\gamma^2(\alpha+R)}{\left\|f(A)\ket{\psi}\right\|}\mathrm{log}\left(\frac{Bl\gamma (1+R/\alpha) }{\|f(A)\ket{\psi}\|\epsilon}\right) \right)$$ queries to $U_A$, $$\mathcal{O}\left(\frac{Bl\gamma (1+R/\alpha) }{\left\|f(A) \ket{\psi}\right\|}\right) $$ queries to $U_{\psi}$ and $$\mathcal{O}\left(\mathrm{log}\left(\frac{(B\gamma+\REV{BK_{\Gamma}}+L)l}{\|f(A)\ket{\psi}\|\epsilon}\right)\right)+a+2$$ ancilla qubits.

\label{theorem:general_complexity}
\end{theorem}

\REV{We note that the query bounds for $U_A$ and $U_\psi$ are independent of the number of quadrature nodes $M$, since a common QSVT sequence processes all discretization points in superposition, and the LCU normalization $\|c\|_1$ admits an $M$-independent bound that controls the required precision and amplitude amplification overhead. 
However, $M$ affects the additional primitive gates needed to prepare the index superposition and implement the index-controlled rotations. 
Here, we briefly discuss the additional gate complexity of the algorithm.} 

\REV{We first consider the outer prepare oracles $C$ and $\widetilde C$. These two oracles prepare superpositions over $M$ basis states with amplitudes proportional to $\sqrt{c_k}$. In general, preparing an arbitrary $M$-dimensional quantum state costs $\mathcal{O}(M)$ gates \cite{SBM06}, which would lead to a gate complexity polynomial in $1/\epsilon$, since $M$ scales polynomially in the inverse precision. However, in our setting, the coefficients $c_k$ are obtained by evaluating known functions at discretized points on the contour. Therefore, using the state preparation methods for efficiently integrable or efficiently computable amplitudes, the circuits for $C$ and $\widetilde C$ can be implemented using only $\mathcal{O}(\operatorname{polylog} M)$ elementary gates \cite{GroverRudolph2002,McardleGilyenBerta2022}. The same discussion applies to the controlled rotations $V$ and $\widetilde V$ used in the block-encoding of $z_k I_N-A$, they can also be implemented with gate complexity $\mathcal{O}(\operatorname{polylog} M)$.}

\REV{Throughout the select oracle $\mathrm{SEL}$, without considering the structure of $U_A$, we will repeat $V$ and $\widetilde V$ and other $\mathcal{O}(a)$-complexity gates like $H$, $\mathrm{e}^{\mathrm{i}\phi_j Z}$ and control gates a total of
$$d=\mathcal{O} \left(\frac{1}{\delta}\mathrm{log}\frac{1}{\epsilon^\prime}\right) = \mathcal{O}\left(\gamma(\alpha+R)\mathrm{log}\left(\frac{Bl\gamma (1+R/\alpha) }{\|f(A)\ket{\psi}\|\epsilon}\right)\right)$$ times. A $\mathrm{SEL}$ combined with $C$ and $\widetilde C$ can form one block-encoding of $f(A)$. Therefore, the additional gate overhead of the block-encoding of $f(A)$ is $$\mathcal{O}\left(\gamma(\alpha+R)\mathrm{log}\left(\frac{Bl\gamma (1+R/\alpha) }{\|f(A)\ket{\psi}\|\epsilon}\right)(a+\operatorname{polylog} M)\right).$$
In addition, we need to apply amplitude amplification to the entire circuit, which involves repeating the circuit $\mathcal{O}\left(\frac{Bl\gamma (1+R/\alpha) }{\left\|f(A)\ket{\psi}\right\|}\right)$ times and using a number of gates that scales as $\mathcal{O}(\log N)$ per round due to the implementation of the reflection. Therefore, by substituting the value of $M$ required in~\cref{eqn:M_choice}, we obtain the additional gate complexity of the algorithm as
$$\mathcal{O}\left(\frac{Bl\gamma^2(1+R/\alpha)(\alpha+R)}{\left\|f(A)\ket{\psi}\right\|}\mathrm{log}\left(\frac{Bl\gamma (1+R/\alpha) }{\|f(A)\ket{\psi}\|\epsilon}\right) \left(a+\log N +\operatorname{polylog}\frac{(B\gamma+BK_{\Gamma}+L)l}{\|f(A)\ket{\psi}\|\epsilon}\right)\right).$$
}

\subsection{Ancilla-efficient algorithm}

In this part, we discuss the complexity of the ancilla-efficient algorithm given in Section \ref{section:Ancilla-efficient Algorithm}. Because  $\widetilde{O}=\ket{0^{a+2}}\bra{0^{a+2}}\otimes O$,
$\ket{\widetilde{\psi}}=\ket{0^{a+2}}\otimes\ket{\psi}$, we have $\|\widetilde{O}\|=\|O\|$ and $U_k\ket{\widetilde{\psi}}=\ket{0^{a+2}}\otimes p^\diamond\left(\frac{z_k I-A}{\alpha+|z_k|}\right)^\dagger\ket{\psi}+\ket{\perp}$.
Following the idea of \cite[Theorem 8]{Chakraborty2024}, we have the following result, whose proof is in Appendix \ref{section:Ancilla-efficient complexity}. 

\begin{theorem}
\label{theorem:Ancilla-efficient_convergence}
Let $\epsilon, \xi_1, \xi_2 \in(0,1)$ be some parameters. $U_k$ is an $(1,a+2,0)$-block-encoding of $p^\diamond\left(\frac{z_k I-A}{\alpha+|z_k|}\right)^\dagger$. Let $\ket{\psi}$ be an initial state and $\ket{\widetilde{\psi}} =\ket{0^{a+2}}\otimes\ket{\psi}$, $O$ be an observable and $\widetilde{O}=\ket{0^{a+2}}\bra{0^{a+2}}\otimes O$.
Suppose that $\left\|f(A)-\sum_j c_j p^\diamond\left(\frac{z_k I-A}{\alpha+|z_k|}\right)^\dagger\right\| \leq \xi_1$, where 
$$
\xi_1 \leq \frac{\epsilon}{6\|O\|\|f(A)\|}.
$$
Furthermore, let
$$
T \geq \frac{8\|O\|^2 \ln ( 2/ \xi_2)\|c\|_1^4}{\epsilon^2}.
$$
Then, Algorithm \ref{algorithm:random_LCU} estimates $\mu$ such that
$$
\left|\mu-\bra{\psi}f(A)^\dagger O f(A)\ket{\psi}\right| \leq \epsilon,
$$
with probability at least $(1-\xi_2)^2$, using of one ancilla qubit and $T$ repetitions of the quantum circuit in Algorithm \ref{algorithm:random_LCU}, Step 3.
\end{theorem}

We choose $M$ and $\epsilon^\prime$ such that 
$$
\frac{(B\gamma^2+\REV{BK_{\Gamma}}\gamma+L\gamma)l^2}{8\pi M}\leq \frac{\xi_1}{2}
$$ 
and 
$$
\left\|\sum_{k=0}^{M-1} c_kp^\diamond\left(\frac{z_k I_N-A}{\alpha + |z_k|}\right)^\dagger -f_M(A)\right\| \leq \frac{\xi_1}{2}.
$$
That means the degree of polynomial $p$ is 
$$
\mathcal{O} \left(\frac{1}{\delta}\mathrm{log\left(\frac{\|c\|_1}{\xi_1}\right)}\right)=\mathcal{O} \left(\gamma(\alpha+R)\mathrm{log}\left(\frac{Bl\gamma(1+R/\alpha)\|O\|\|f(A)\|}{\epsilon}\right) \right).
$$ 
The degree of polynomial $p$ is the same as the number of times QSVT visits $U_A$.
For each time we need $a$ ancilla qubits for $U_A$, $1$ for inner layer LCU, $1$ for QSVT and $1$ for outer layer LCU. Therefore, we have the following theorem:

\begin{theorem}
Suppose that $\Gamma$ is a piecewise $C^{1,1}$ simple closed curve covering
all eigenvalues of $A$ on the disk $|z| \leq R$, and $f(z)$ is an holomorphic function inside $\Gamma$ and $L$-Lipschitz continuous on $\Gamma$. Let $B$ be the maximum value of $|f(z)|$ on $\Gamma$. $U_A$ is an $(\alpha,a,0)$-block-encoding of A. Under the assumption that for any $z \in \Gamma$, $\left\|(z I_N -A)^{-1}\right\| \leq \gamma$, let
$$
T =\mathcal{O} \left(\frac{\|O\|^2 \ln (1 / \xi_2)(Bl\gamma)^4(1+R/\alpha)^4}{\epsilon^2 }\right),
$$ we can solve the problem defined by Definition \ref{definition:estimation_problem} with probability at least $(1-\xi_2)^2$ using a quantum circuit with
$$
\mathcal{O} \left(\gamma(\alpha+R)\mathrm{log}\left(\frac{Bl\gamma(1+R/\alpha)\|O\|\|f(A)\|}{\epsilon}\right)\right)
$$ queries to $U_A$, $1$ query to $U_{\psi}$ and $a+3$ ancilla qubits, and repeat this circuit $T$ times.

\label{theorem:Ancilla-efficient_complexity}
\end{theorem}

\section{Applications}\label{sec:applications}

In this section, we will introduce applications of the contour-integral-based algorithm and show their complexities. 
\subsection{Linear ordinary differential equations}

\subsubsection{Generic scaling}\label{sec:app_ODE_generic}

Next we consider the initial value problem of the system of homogeneous linear ordinary differential equations (ODEs) 
\begin{equation}
    \frac{\mathrm{d}}{\mathrm{d}t}x(t)=A x(t), \quad x(0)=\ket{\psi}.
    \label{application:ODE}
\end{equation}
The solution is
\begin{equation}
    x(t)=\mathrm{e}^{At}\ket{\psi}.
\end{equation}
Let $T$ be the total evolution time, and we should implement $\mathrm{e}^{AT}$. We hope that this equation has bounded solutions in the long-term evolution process. Thus, we assume that the ODE is asymptotically stable, which means all eigenvalues of $A$ have \REV{non-positive} real parts. Suppose the spectral radius of $A$ is $\rho(A)$, then all the eigenvalues of $A$ are covered by the semicircle $\{z:|z|\leq\rho(A), \mathrm{Re}(z)\leq0\}$. We choose $\Gamma$ as the boundary of $\{z:|z|\leq\rho (A)+a, \mathrm{Re}(z)\leq a\}$ for a positive real number $a$, as shown in Fig.\ref{fig:spectrum_contour_1}, left. 

\begin{figure}[ht]
    \centering
    \begin{tikzpicture}[>=Stealth, scale=1.2]
        \draw[->, thick] (-1.5,0) -- (1.5,0) node[right] {$\mathrm{Re}$};
        \draw[->, thick] (0,-1.5) -- (0,1.5) node[above] {$\mathrm{Im}$};
        
        \draw[blue, thick] (0.1,1.1) arc (85:275:1.1);
        \draw[blue, thick] (0.1,1.1) -- (0.1,-1.1);
        \filldraw[fill=red!20, draw=red, thick] 
        (0,1) arc (90:270:1) -- cycle;
        \node[right] at (0.02, 0.88)  {$R_2 =\rho(A)$};
        \node[right] at (0.02, 1.22)  {$R_1 =\rho(A)+a$};
        \node[right] at (0.3, -0.6) {$f(z)=\mathrm{e}^{Tz}$};
        \node[right] at (0.3, -1.0) {$f(z)=\varphi_k(z)$};
    \end{tikzpicture}
    \begin{tikzpicture}[>=Stealth, scale=1.2]
        \draw[->, thick] (-1.5,0) -- (1.5,0) node[right] {$\mathrm{Re}$};
        \draw[->, thick] (0,-1.5) -- (0,1.5) node[above] {$\mathrm{Im}$};
        
        \draw[blue, thick] (0,1.1) arc (90:270:1.1);
        \draw[blue, thick] (0,1.1) -- (0,-1.1);
        \filldraw[fill=red!20, draw=red, thick] 
        (-0.1,1) arc (95:265:1) -- cycle;
        \node[right] at (0.02, 0.88)  {$R_2 =\rho(A)$};
        \node[right] at (0.02, 1.22)  {$R_1 =\rho(A)+a$};
        \node[right] at (0.3, -0.8) {$f(z)=e^{Tz}$};
    \end{tikzpicture}
    \begin{tikzpicture}[>=Stealth, scale=1.2]
    
        \draw[->, thick] (-1.5,0) -- (1.5,0) node[right] {$\mathrm{Re}$};
        \draw[->, thick] (0,-1.5) -- (0,1.5) node[above] {$\mathrm{Im}$};
        
        \draw[blue, thick] (0,0) circle [radius=1.1];
        \filldraw[fill=red!20, draw=red, thick] 
        (0,0) circle [radius=1];
        \node[right] at (0.9, 0.88)  {$R_2 =\rho(A)$};
        \node[right] at (0.9, 1.22)  {$R_1 =\rho(A)+a$};
    \end{tikzpicture}
    \caption{The spectrum of the matrix A and the corresponding covering contour under different assumptions. Left: the real part of the spectrum of $A$ does not exceed $0$. Middle: the real part of the spectrum of $A$ has an upper bound $-a < 0$. Right: no special assumptions.}
    \label{fig:spectrum_contour_1}
\end{figure}

Assume that $A$ can be diagonalized, let $A = SDS^{-1}$, then $
\left\|(z I -A)^{-1}\right\| \leq \frac{\kappa_S}{a}$, where $\kappa_S$ is the condition number of $S$. Therefore, we can let $\gamma=\frac{\kappa_S}{a}$, and $\Gamma$ is on the disk $|z| \leq R = \rho (A)+a$. \REV{The maximum value of $|f(z)|$ on $\Gamma$ is $B = e^{aT}$, and the upper bound on the tangent variation rate of $\Gamma$ is $K_{\Gamma}=\frac{1}{\rho
(A)+a}\leq\frac{\kappa_s}{a}$.} Then, according to Theorem \ref{theorem:general_complexity}, we need $$\mathcal{O}\left(\frac{\mathrm{e}^{aT}\kappa_S^2(\rho (A)+a)(\alpha+\rho (A)+a)}{a^2\left\|x(T)\right\|}\mathrm{log}\left(\frac{\mathrm{e}^{aT}\kappa_S(\rho (A)+a)}{\|x(T)\|a\epsilon}\right)  \right)$$ queries to $U_A$, $$\mathcal{O}\left(\frac{\mathrm{e}^{aT}\kappa_S(\rho (A)+a)}{a\left\|x(T)\right\|}\right) $$ queries to $U_{\psi}$ and $$\mathcal{O}\left(\mathrm{log}\left(\frac{(B\kappa_S/a+L)l}{\|x(T)\|\epsilon}\right)\right)+\tilde{a}+2$$ ancilla qubits to implement $\mathrm{e}^{AT}$. Suppose that $T>1$ and choose $a = \frac{1}{T}$, note that $\rho(A)\leq\|A\|\leq\alpha$, then we need $$\mathcal{O}\left(\frac{\kappa_S^2\alpha^2T^2}{\left\|x(T)\right\|}\mathrm{log}\left(\frac{\kappa_S\alpha T}{\|x(T)\|\epsilon}\right)  \right)$$ queries to $U_A$ , $$\mathcal{O}\left(\frac{\kappa_S\alpha T}{\left\|x(T)\right\|}\right) $$ queries to $U_{\psi}$ and $$\mathcal{O}\left(\mathrm{log}\left(\frac{B\kappa_S\alpha T}{\|x(T)\|\epsilon}\right)\right)+\tilde{a}+2$$ ancilla qubits. 

A comparison is given in Table \ref{compare:ODE}. These methods assume the ODE to be semi-dissipative from different perspectives. Some methods assume $\mathrm{Re}(\lambda_A)\leq 0$, which means that all eigenvalues of $A$ have \REV{non-positive} real parts. While others assume $A+A^\dagger\preceq 0$ is a negative semi-definite matrix. In fact, the condition $A+A^\dagger\preceq 0$ is stronger than $\mathrm{Re}(\lambda_A)\leq 0$, see Appendix \ref{lemma:Dissipative_conditions}. When using assumption $\mathrm{Re}(\lambda_A)\leq 0$, we also assume $A$ to be diagonalizable, but this is only for complexity analysis. Not assuming $A$ to be diagonalizable does not affect the correctness of the method, but it is difficult to give a good complexity estimate.

\begin{table}[ht]
\centering
\setlength{\tabcolsep}{3pt}
\renewcommand{\arraystretch}{1.8}
\scalebox{0.9}{
\begin{tabular}{c|c|c}
\hline\hline
\textbf{Methods} & \textbf{Condition} & \textbf{Query complexities} \\
\hline Taylor \cite{BerryChildsOstranderWang2017} & $\mathrm{Re}(\lambda_A)\leq 0$ & $\widetilde{\mathcal{O}}\left(\frac{\kappa_S \alpha T \max_t\|x(t)\|}{{\left\|x(T)\right\|}} \operatorname{poly}\left(\log \left(\frac{1}{\epsilon}\right)\right)\right)$ \\
\hline Improved Taylor \cite{Krovi2022} & $\max_t \|e^{At}\| \leq C$ & $\widetilde{\mathcal{O}}\left(\frac{ C \alpha T  \max_t\|x(t)\|}{{\left\|x(T)\right\|}}\operatorname{poly}\left(\log \left(\frac{1}{\epsilon}\right)\right)\right)$ \\
\hline Time-marching \cite{FangLinTong2023} & $A+A^\dagger\preceq0$ & $\widetilde{\mathcal{O}}\left(\frac{\alpha^2 T^2}{{\left\|x(T)\right\|}} \left(\log \left(\frac{1}{\epsilon}\right)\right)^2 \right)$  \\
\hline Time-marching \cite{FangLinTong2023} & $\mathrm{Re}(\lambda_A)\leq 0$ & $\widetilde{\mathcal{O}}\left(\frac{\alpha^2 T^2\kappa_S^{\mathcal{O}(T)}}{{\left\|x(T)\right\|}} \left(\log \left(\frac{1}{\epsilon}\right)\right)^2 \right)$ \\
\hline LCHS \cite{AnChildsLin2023}& $A+A^\dagger\preceq0$ & $\widetilde{\mathcal{O}}\left(\frac{\alpha T}{{\left\|x(T)\right\|}}\left(\log \left(\frac{1}{\epsilon}\right)\right)^{1 + o(1)}\right)$  \\
\hline Approximate LCHS \cite{LowSomma2025}& $A+A^\dagger\preceq0$ & $\widetilde{\mathcal{O}}\left(\frac{\alpha T}{{\left\|x(T)\right\|}}\log \left(\frac{1}{\epsilon}\right)\right)$  \\
\hline Block preconditioning
\cite{LowSu2024quantumlinearalgorithmoptimal} & $A+A^\dagger\preceq0$ & $\widetilde{\mathcal{O}}\left(\frac{\alpha T}{\left\|x(T)  \right\|}(\log\frac{1}{\epsilon})^2\right)$ \\
\hline Block preconditioning
\cite{LowSu2024quantumlinearalgorithmoptimal} & $\mathrm{Re}(\lambda_A)\leq 0$ & $\widetilde{\mathcal{O}}\left(\frac{\kappa_s^2\alpha T}{\left\|x(T)  \right\|}(\log\frac{1}{\epsilon})^2\right)$ \\
\hline This work (Section~\ref{sec:app_ODE_generic}) & $\mathrm{Re}(\lambda_A)\leq 0$ & $\widetilde{\mathcal{O}}\left(\frac{\kappa_S^2\alpha^2T^2}{\left\|x(T)\right\|}\mathrm{log}(\frac{1}{\epsilon})  \right)$ \\
\hline\hline 
Fast-forwarded Taylor \cite{JenningsLostaglioLowrieEtAl2024} & $\mathrm{Re}(\lambda_A) < 0$ & $\widetilde{\mathcal{O}}\left( \frac{\kappa_S^{1/2} \alpha T^{3/4} \max_t\|x(t)\|}{\|x(T)\|} \left(\log\left(\frac{1}{\epsilon}\right)\right)^2\right)$ \\
\hline Fast-forwarded Taylor \cite{AnOnwuntaYang2024} & $A + A^{\dagger} \prec 0 $  & $ \widetilde{\mathcal{O}} \left( \frac{\alpha T^{1/2} \max_t\|x(t)\|}{\|x(T)\|} \left(\log\left(\frac{1}{\epsilon}\right)\right)^2 \right)  $ \\
\hline Fast-forwarded LCHS \cite{YangOnwuntaAn2025} & $A + A^{\dagger} \prec 0 $ & $\widetilde{\mathcal{O}}\left( \frac{\alpha}{\|x(T)\|} \left(\log\left(\frac{1}{\epsilon}\right)\right)^{2+o(1)} \right)  $ \\
\hline This work (Section \ref{sec:app_ODE_FF}) & $\mathrm{Re}(\lambda_A) < 0$ & $\widetilde{\mathcal{O}}\left( \frac{\kappa_S^2 \alpha^2 }{\|x(T)\|} \log\left(\frac{1}{\epsilon}\right) \right) $ \\
\hline \hline 
\end{tabular}
}
\caption{Comparison among contour integral and previous methods for homogeneous ODEs. Here $\alpha \geq\|A\|, T$ is the evolution time, and $\epsilon$ is the error. When the condition $\mathrm{Re}(\lambda_A)\leq 0$ is used, $A$ is further assumed to be diagonalizable with matrix $S$ such that $\kappa_S \geq\|S\|\left| S^{-1}\right\|$. \REV{We treat other unspecified parameters as constants. The notation $\widetilde{\mathcal{O}}$ hides the nondominant poly-logarithmic terms, i.e., $\widetilde{\mathcal{O}}(c(x)) = \mathcal{O}(c(x) \operatorname{polylog}(c(x)))$ for a scaling function $c(x)$. } }
\label{compare:ODE}
\end{table}

Compared with the method based on hypothesis $A+A^\dagger\preceq0$, the contour integral method relaxes the assumption. Compared to the method based on hypothesis $\mathrm{Re}(\lambda_A)\leq 0$, although the contour integral method has a worse dependence on $T$, it has a better dependence on $\epsilon$. Therefore, the contour integral method has an advantage in short-term evolution with high precision requirements.

If we use the ancilla-efficient algorithm, suppose that we would observe $x(T)$ under $O$, then according to Theorem \ref{theorem:Ancilla-efficient_complexity}, each time we need
$$
\mathcal{O} \left(\kappa_S\alpha T\mathrm{log}\left(\frac{\kappa_S\alpha T\|O\|\|e^{AT}\|}{\epsilon}\right)\right)
$$ queries to $U_A$, 
$\mathcal{O}(K)$ queries to $U_{\psi}$ and $\tilde{a}+3$ ancilla qubits, and we need to repeat the circuit $K$ times, where
$$
K =\mathcal{O} \left(\frac{\|O\|^2 \ln (1 / \xi_2)(\kappa_S\alpha T)^4}{\epsilon^2 }\right).
$$

\subsubsection{Fast-forwarding}\label{sec:app_ODE_FF}
Consider the problem
\begin{equation}
    \frac{\mathrm{d}}{\mathrm{d}t}x(t)=A x(t)+b, \quad x(0)=x_0.
    \label{application:ODE_inhomogeneous}
\end{equation}
\REV{When $A$ is invertible,} the solution is
\begin{equation}
    x(t)=\mathrm{e}^{At}(x_0+A^{-1}b)-A^{-1}b.
\end{equation}
Let $\psi=x_0+A^{-1}b$, then we can first compute $\mathrm{e}^{At}\ket{\psi}$ then use LCU to sum up $-A^{-1}b$. In some cases, the real part of the eigenvalues of $A$ has an upper bound $-a$ less than $0$. Then all the eigenvalues of $A$ are covered by the semicircle $\{z:|z|\leq\rho(A), \mathrm{Re}(z)\leq-a\}$, \REV{which implies that $A$ is invertible.} We choose $\Gamma$ as the boundary of $\{z:|z|\leq\rho (A)+a, \mathrm{Re}(z)\leq 0\}$, as shown in Fig.\ref{fig:spectrum_contour_1}, middle. 

For simplicity, we only consider the complexity when $b=0$ and do not take into account the additional cost of using LCU to sum up $-A^{-1}b$. The generalization for non-homogeneous case is also straightforward, and we will not go into the calculations. Assume that $A$ can be diagonalized, let $A = SDS^{-1}$, then $
\left\|(z I -A)^{-1}\right\| \leq \frac{\kappa_S}{a}$, where $\kappa_S$ is the condition number of $S$. Therefore, we can let $\gamma=\frac{\kappa_S}{a}$, and $\Gamma$ is on the disk $|z| \leq R = \rho (A)+a$. Since the maximum value of $|f|$ on $\Gamma$ is $1$, \REV{and the upper bound on the tangent variation rate of $\Gamma$ is $K_{\Gamma}=\frac{1}{\rho
(A)+a}\leq\frac{\kappa_s}{a}$,} according to Theorem \ref{theorem:general_complexity}, we need $$\mathcal{O}\left(\frac{\kappa_S^2(\rho (A)+a)(\alpha+\rho (A)+a)}{a^2\left\|x(T)\right\|} \mathrm{log}\left(\frac{\kappa_S(\rho (A)+a)}{\|x(T)\|a\epsilon}\right)\right)$$ queries to $U_A$, $$\mathcal{O}\left(\frac{\kappa_S(\rho (A)+a)}{a\left\|x(T)\right\|}\right) $$ queries to $U_{\psi}$ and $$\mathcal{O}\left(\mathrm{log}\left(\frac{(B\kappa_S/a+L)l}{\|x(T)\|\epsilon}\right)\right)+\tilde{a}+2$$ ancilla qubits to implement $x(T)$. 

Note that when using this method, the computational cost is no longer explicitly dependent on the time $T$. 
However, we remark that there is still implicit $T$ dependence in the norm $\|x(T)\|$ and, even worse, the solution norm for the homogeneous ODEs decays asymptotically exponentially in the large $T$ regime when the real parts of the eigenvalues of $A$ are negative. 
So the contour integral algorithm is more practically relevant in the moderate time regime. 
Furthermore, fortunately, for inhomogeneous ODEs, the solution norm does not decay exponentially in general, and the contour integral algorithm becomes advantageous for any time regime.

If we use the ancilla-efficient algorithm, suppose that we would observe $x(T)$ under $O$, then according to Theorem \ref{theorem:Ancilla-efficient_complexity}, each time we need
$$
\mathcal{O} \left(\frac{\kappa_S\alpha}{a}\mathrm{log}\left(\frac{\kappa_S\alpha\|O\|\|f(A)\|}{a\epsilon}\right)\right)
$$ queries to $U_A$, 
$1$ queries to $U_{\psi}$ and $\tilde{a}+3$ ancilla qubits, and we need to repeat the circuit $K$ times, where
$$
K =\mathcal{O} \left(\frac{\|O\|^2 \ln (1 / \xi_2)(\kappa_S\alpha/a)^4}{\epsilon^2}\right).
$$
\subsection{\texorpdfstring{$\varphi$}{phi}-functions} 
\label{sec:app_phi_functions}

\REV{Consider the $\varphi$-functions
\begin{equation}
    \varphi_k(z)
    =
    \sum_{j=0}^{\infty}\frac{z^j}{(j+k)!}
    =
    \frac{1}{(k-1)!}
    \int_0^1 \mathrm{e}^{(1-\theta)z}\theta^{k-1}\mathrm{d}\theta,
    \qquad k\geq 1,
    \label{eq:phi_function_definition}
\end{equation}
and let $\varphi_0(z)=\mathrm{e}^z$, our target is to approximate the normalization version of $\varphi_k(A)\ket{\psi}$.}

\REV{Since the $\varphi$-function is relatively stable in the left half-plane and grows exponentially in the right half-plane, assume $A=SDS^{-1}$ is diagonalizable and all its eigenvalues
have non-positive real parts. Suppose the spectral radius of $A$ is
$\rho(A)$, and choose $\Gamma$ as the boundary of $\{z:|z|\leq \rho(A)+a,\ \operatorname{Re}(z)\leq a\}$, for some $a>0$, as shown in Fig.~\ref{fig:spectrum_contour_1}, left.
Then $\|(zI-A)^{-1}\|\leq \frac{\kappa_S}{a}$,
and therefore we can let $\gamma=\frac{\kappa_S}{a}$. Moreover, for $\operatorname{Re}(z)\leq a$,
\begin{equation}
    |\varphi_k(z)|
    \leq
    \frac{1}{(k-1)!}
    \int_0^1
    \mathrm{e}^{(1-\theta)a}\theta^{k-1}\mathrm{d}\theta 
    \leq
    \frac{\mathrm{e}^a}{k!}.
\label{eq:varphi-bound}
\end{equation}
Thus, we can take
\[
    B=\frac{\mathrm{e}^a}{k!},
    \qquad
    R=\rho(A)+a,
    \qquad
    K_\Gamma=\frac{1}{\rho(A)+a}
    \leq \frac{\kappa_S}{a}.
\]}

\REV{According to Theorem~\ref{theorem:general_complexity}, the contour integral algorithm requires
\begin{equation}
    \mathcal{O}\left(
    \frac{
        \mathrm{e}^a\kappa_S^2(\rho(A)+a)(\alpha+\rho(A)+a)
    }{
        k!a^2\|\varphi_k(A)|\psi\rangle\|
    }
    \mathrm{log}\left(
    \frac{
        \mathrm{e}^a\kappa_S(\rho(A)+a)
    }{
        k!a\|\varphi_k(A)|\psi\rangle\|\epsilon
    }
    \right)
    \right)
\end{equation}
queries to $U_A$,
\begin{equation}
    \mathcal{O}\left(
    \frac{
        \mathrm{e}^a\kappa_S(\rho(A)+a)
    }{
        k!a\|\varphi_k(A)|\psi\rangle\|
    }
    \right)
\end{equation}
queries to $U_\psi$, and
\begin{equation}
    \mathcal{O}\left(
    \mathrm{log}\left(
    \frac{
        (\frac{\mathrm{e}^a\kappa_S}{k!\cdot a}+L)l
    }{
        \|\varphi_k(A)|\psi\rangle\|\epsilon
    }
    \right)
    \right)
    +\widetilde{a}+2
\end{equation}
ancilla qubits.}

\REV{If we use the ancilla-efficient algorithm, according to Theorem~\ref{theorem:Ancilla-efficient_complexity},
each run of the quantum circuit requires
\begin{equation}
    \mathcal{O}\left(
    \frac{\kappa_S}{a}
    (\alpha+\rho(A)+a)
    \mathrm{log}\left(
    \frac{
        \mathrm{e}^a\kappa_S(\rho(A)+a)
        \|O\|\|\varphi_k(A)\|
    }{
        k!a\epsilon
    }
    \right)
    \right)
\end{equation}
queries to $U_A$, one query to $U_\psi$, and
$\widetilde{a}+3$ ancilla qubits, and the circuit is repeated
\begin{equation}
    K
    =
    \mathcal{O}\left(
    \frac{
        \|O\|^2\ln(1/\xi_2)
    }{
        \epsilon^2
    }
    \left(
    \frac{
        \mathrm{e}^a\kappa_S(\rho(A)+a)
    }{
        k!a
    }
    \right)^4
    \right)
\end{equation}
times.}

\REV{An important application of $\varphi$-functions is the inhomogeneous
linear ODE
\begin{equation}
    \frac{\mathrm{d}}{\mathrm{d}t}x(t)=Ax(t)+b,
    \qquad
    x(0)=x_0,
    \label{application:ODE_phi_constant}
\end{equation}
whose solution is
\begin{equation}
    x(T)
    =
    \mathrm{e}^{AT}x_0
    +
    T\varphi_1(TA)b.
    \label{eq:inhomogeneous_phi_solution}
\end{equation}
Therefore, the inhomogeneous term can be implemented by applying the
above result to the matrix $TA$ with $k=1$.}

\REV{When $A$ is invertible,
\begin{equation}
    T\varphi_1(TA)
    =
    A^{-1}(\mathrm{e}^{AT}-I),
\end{equation}
and we can alternatively implement the inhomogeneous term by separately
implementing $\mathrm{e}^{AT}$ and $A^{-1}$. In the fast-forwarding regime
considered in Section \ref{sec:app_ODE_FF}, where $A=SDS^{-1}$ is diagonalizable and
$\operatorname{Re}(\lambda_A)\leq-a<0$,
\begin{equation}
    \|A^{-1}\|
    \leq
    \frac{\kappa_S}{a},
\end{equation}
so $A^{-1}$ can be efficiently approximated by QSVT.}

\REV{On the other hand, when
$\operatorname{Re}(\lambda_A)\leq0$, $A$ may be singular and
$A^{-1}$ may be ill-conditioned or undefined. Since
\begin{equation}
    \varphi_1(z)
    =
    \frac{\mathrm{e}^z-1}{z},
    \qquad
    \varphi_1(0)=1,
\end{equation}
$\varphi_1$ is entire, and the contour integral method provides a
direct way to implement $T\varphi_1(TA)b$ in this case.}

\REV{To apply the above complexity to the inhomogeneous ODE in
Eq. \eqref{application:ODE_phi_constant}, we take $k=1$ and replace $A$ by
$TA$. In this case, the block-encoding normalization becomes $T\alpha$
and $\rho(TA)=T\rho(A)$. Choosing the contour parameter to be $1$,
which is equivalent to choosing $a=1/T$ for the original matrix $A$,
and noting that $\rho(A)\leq\|A\|\leq\alpha$, the query complexity for
implementing $\varphi_1(TA)|b\rangle$ becomes
\begin{equation}
    \mathcal{O}\left(
    \frac{\kappa_S^2\alpha^2T^2}
    {\|\varphi_1(TA)|b\rangle\|}
    \mathrm{log}\left(
    \frac{\kappa_S\alpha T}
    {\|\varphi_1(TA)|b\rangle\|\epsilon}
    \right)
    \right)
\end{equation}
queries to $U_A$. Therefore, the inhomogeneous contribution
$T\varphi_1(TA)b$ in Eq. \eqref{eq:inhomogeneous_phi_solution} can be
implemented with the same query complexity, since the scalar factor
$T$ does not affect the normalized output state.}

\subsection{Fr\'echet derivatives of matrix functions}
\label{sec:app_Frechet_derivatives}

\REV{Consider the Fr\'echet derivative of a matrix function. 
Let $f$ be holomorphic in $\Gamma$ containing the spectrum of $A$, and let $E$ be a perturbation matrix. 
The Fr\'echet derivative of $f$ at $A$ in the direction $E$ is defined as
\begin{equation}
    L_f(A,E)
    =
    \left.
    \frac{\mathrm d}{\mathrm d s}f(A+sE)
    \right|_{s=0}.
    \label{eq:Frechet_derivative_definition}
\end{equation}}

\REV{By differentiating the Cauchy integral representation of $f(A)$, we obtain
\begin{equation}
    L_f(A,E)
    =
    \frac{1}{2\pi \mathrm{i}}
    \oint_{\Gamma}
    f(z)
    (zI-A)^{-1}
    E
    (zI-A)^{-1}
    \mathrm d z,
    \label{eq:Frechet_derivative_contour}
\end{equation}
where $\Gamma$ is a contour covering all eigenvalues of $A$. 
This formula shows that the Fr\'echet derivative can be implemented by almost the same contour integral method as $f(A)$, except that two resolvents and one perturbation matrix $E$ appear in the integrand.}

\REV{A convenient way to reduce this problem to the matrix function problem studied in this work is to introduce the block upper triangular matrix
\begin{equation}
    \mathcal A
    =
    \begin{pmatrix}
        A & E \\
        0 & A
    \end{pmatrix}
    \label{eq:Frechet_block_matrix}
\end{equation}
which can be implemented by LCU through the formula
\begin{equation}
    \mathcal A
    =
    I_2\otimes A
    +
    \frac{1}{2}X\otimes E
    +
    \frac{1}{2}\mathrm iY\otimes E.
\end{equation}
Suppose $U_A$, $U_E$ are the block-encodings with $\widetilde{a}$ ancilla qubits and block-encoding factors $\alpha_A$, $\alpha_E$, then this provides an $(\alpha_A+\alpha_E,\widetilde{a}+2)$-block-encoding of $\mathcal{A}$. Since
\begin{equation}
\begin{aligned}
    f(\mathcal A)
    =&\frac{1}{2 \pi \mathrm{i}} \oint_{\Gamma} f(z)\left(z I-\mathcal A\right)^{-1} \mathrm{~d} z=\frac{1}{2 \pi \mathrm{i}} \oint_{\Gamma} f(z)\begin{pmatrix}
        zI-A & -E \\
        0 & zI-A
    \end{pmatrix}^{-1} \mathrm{~d} z\\
    =&\frac{1}{2 \pi \mathrm{i}} \oint_{\Gamma} f(z)\begin{pmatrix}
        (zI-A)^{-1} & (zI-A)^{-1}E(zI-A)^{-1} \\
        0 & (zI-A)^{-1}
    \end{pmatrix} \mathrm{~d} z=
    \begin{pmatrix}
        f(A) & L_f(A,E) \\
        0 & f(A)
    \end{pmatrix}.
    \label{eq:Frechet_block_function}
\end{aligned}
\end{equation}
Therefore, preparing $L_f(A,E)\ket{\psi}$ can be reduced to implementing $f(\mathcal A)$ on a block-encoded enlarged matrix and then postselecting the off-diagonal block. Let $U_{f(\mathcal A)}$ be the block-encoding of $f(\mathcal{A})/\alpha$, since
\begin{equation}
    U_{f(\mathcal A)}(I_\mathcal{A}\otimes X\otimes I_N)\ket{0_\mathcal{A}}\ket{0}\ket{\psi}= \frac{1}{\alpha}\ket{0_\mathcal{A}}f(\mathcal A)\ket{1}\ket{\psi}+\ket{\perp}= \frac{1}{\alpha}\ket{0_\mathcal{A}}\ket{0}L_f(A,E)\ket{\psi}+\ket{\perp},
\end{equation}
one additional $X$ gate can convert this construction into a strict block-encoding of $L_f(A,E)$. In the construction, we use one ancilla qubit to extend $L_f(A,E)$ to $\mathcal{A}$, and two ancilla qubits to implement $\mathcal{A}$ using LCU.}

\REV{We now estimate the complexity. 
Assume that $A$ is diagonalizable, $A=SDS^{-1}$, then $\left\|(zI-A)^{-1}\right\|\leq \frac{\kappa_S}{a}$, and
\begin{equation}
    \left\|(zI-\mathcal A)^{-1}\right\|=\left\|\begin{pmatrix}
        (zI-A)^{-1} & (zI-A)^{-1}E(zI-A)^{-1} \\
        0 & (zI-A)^{-1}
    \end{pmatrix}\right\|
    \leq
    \frac{2\kappa_S}{a}+\|E\|\left(\frac{\kappa_S}{a}\right)^2:=\gamma_{\mathcal{A}}.
    \label{eq:Frechet_block_resolvent_bound}
\end{equation}
Then according to Theorem \ref{theorem:general_complexity} we need $$\mathcal{O}\left(\frac{Bl\gamma_{\mathcal{A}}^2(\alpha_A+\alpha_E)}{\left\|L_f(A,E)\ket{\psi}\right\|}\mathrm{log}\left(\frac{Bl\gamma_{\mathcal{A}} }{\|L_f(A,E)\ket{\psi}\|\epsilon}\right) \right)$$ queries to $U_A$ and $U_E$, $$\mathcal{O}\left(\frac{Bl\gamma_{\mathcal{A}}}{\left\|L_f(A,E) \ket{\psi}\right\|}\right) $$ queries to $U_{\psi}$ and $$\mathcal{O}\left(\mathrm{log}\left(\frac{(B\gamma_{\mathcal{A}}+L)l}{\|L_f(A,E)\ket{\psi}\|\epsilon}\right)\right)+\widetilde{a}+5$$ ancilla qubits.}

\REV{If we use the ancilla-efficient algorithm, suppose that we would observe $x(T)$ under $O$, then according to Theorem \ref{theorem:Ancilla-efficient_complexity}, each time we need
$$
\mathcal{O} \left(\gamma_{\mathcal{A}}(\alpha_A+\alpha_E)\mathrm{log}\left(\frac{Bl\gamma_{\mathcal{A}}\|O\|\|L_f(A,E)\|}{\epsilon}\right)\right)
$$ queries to $U_A$ and $U_E$, 
$1$ queries to $U_{\psi}$ and $\tilde{a}+6$ ancilla qubits, and we need to repeat the circuit $K$ times, where
$$
K=\mathcal{O} \left(\frac{\|O\|^2 \ln (1 / \xi_2)(Bl\gamma_{\mathcal{A}})^4}{\epsilon^2 }\right).
$$}

\subsection{Matrix polynomial}

Consider the problem of implementing a general matrix polynomial.  Suppose that we hope to implement $f(A)$, while $f$ is a polynomial and $A$ is a $N \times N$ matrix. Suppose the spectral radius of $A$ is $\rho(A)$, then all the eigenvalues of $A$ are covered by the circle $\{z:|z|\leq\rho(A) \}$. We choose $\Gamma$ as the boundary of $\{z:|z|\leq\rho (A)+a \}$, as shown in Fig.\ref{fig:spectrum_contour_1}, right. 

Assume that $A$ can be diagonalized as $A = SDS^{-1}$ and each element on the diagonal of $zI-D$ is not less than $a$, then \REV{a straightforward estimate yields}
\begin{equation}
    \left\|(z I -A)^{-1}\right\| = \left\|S(z I -D)^{-1}S^{-1}\right\| \leq  \left\|S\right\|\left\|S^{-1}\right\|\left\|(z I -D)^{-1}\right\| \leq \frac{\kappa_S}{a}, 
\end{equation}
where $\kappa_S =\|S\|\|S^{-1}\|$ is the condition number of $S$. 
\REV{However, for the circular contour here, when $R \geq \|A\|$ we can obtain a sharper estimate. 
Let $R = \rho(A) + a$ be the radius of $\Gamma$.  
Notice that $\|A^j\|/R^j \leq \min\{1, \kappa_S\rho(A)^j/R^j$\}, and we can bound for $|z|=R$ that 
\begin{equation}
    \|(zI-A)^{-1}\| = |z|^{-1}\|\sum_{j=0}^{\infty} (A/z)^j\| \leq \frac{1}{R} \sum_{j=0}^{\infty} \frac{\|A^j\|}{R^j} \leq \frac{1}{R} \sum_{j=0}^{\infty} \min\left\{1, \frac{\kappa_S\rho(A)^j}{R^j}\right\} \leq \frac{1}{R}\lceil \frac{\log\kappa_S}{\log(R/\rho(A))} \rceil +\frac{1}{a}, 
\end{equation}
where the last inequality is obtained by bounding the first $\lceil \frac{\log\kappa_S}{\log(R/\rho(A))} \rceil$ terms by $1$ and the remaining series geometrically. 
Therefore, we can let $\gamma = \mathcal{O}\left(\frac{\log \kappa_S}{a\log(R/\rho(A))} \right)$, }
and $\Gamma$ is on the disk $\{z:|z| \leq R = \rho (A)+a\}$. \REV{The upper bound on the tangent variation rate of $\Gamma$ is $K_{\Gamma}=\frac{1}{\rho (A)+a}\leq \frac{\kappa_S}{a}$,} and assume that the maximum value of $f(z)$ on $\Gamma$ is $B$. Then according to Theorem \ref{theorem:general_complexity}, we need 
\REV{$$\mathcal{O}\left(\frac{B  R(\alpha+R) \log^2 \kappa_S }{a^2 \left\|f(A)\ket{\psi}\right\| \log^2(R/\rho(A)) }\mathrm{log}\left(\frac{B R \log \kappa_S}{\epsilon \|f(A)\ket{\psi}\| a \log(R/\rho(A))} \right) \right)$$ }
queries to $U_A$, 
\REV{$$\mathcal{O}\left(\frac{B R \log \kappa_S }{\left\|f(A) \ket{\psi}\right\| a \log(R/\rho(A))}\right) $$ }
queries to $U_{\psi}$ and $$\mathcal{O}\left(\mathrm{log}\left(\frac{(B\kappa_S/a+L)l}{\|f(A) \ket{\psi}\|\epsilon}\right) \right)+\tilde{a}+2$$ ancilla qubits to implement $f(A)$. Because $\rho (A)\leq\|A\|\leq\alpha$, we can choose $a \sim \alpha$ so that we need 
\REV{$$\mathcal{O}\left(\frac{B \log^2 \kappa_S}{\left\|f(A)\ket{\psi}\right\|}\mathrm{log}\left(\frac{B\log\kappa_S}{\|f(A)\ket{\psi}\|\epsilon}\right)  \right)$$ }
queries to $U_A$, 
$$\mathcal{O}\left(\frac{B\log\kappa_S}{\left\|f(A) \ket{\psi}\right\|}\right) $$
queries to $U_{\psi}$ and $$\mathcal{O}\left(\mathrm{log}\left(\frac{B\kappa_S+\alpha L}{\alpha\|f(A) \ket{\psi}\|\epsilon}\right) \right)+\tilde{a}+2$$ ancilla qubits. It is only related to the maximum value of $f$ in the disk, and has nothing to do with the degree of $f$ or the coefficient of the monomial in $f$. 

If we want a naive implementation of matrix polynomials, we might think of using a block-encoding multiplication for each monomial. However, if the block-encoding of the matrix is directly multiplied, the information in state $\ket{\perp}$ will be returned to the target state, making the calculation result wrong. Therefore, we need to use separate ancilla qubits for each block-encoding. This means that even if we only calculate $f(x) = x^n$ and $U_A$ is an $(\alpha,a,0)$-block-encoding of $A$, then for $A^n$ the number of ancilla qubits required is  $an$, and the probability of success is $\alpha^n$. Even with amplitude amplification, the query complicity is still exponentially depends of $n$.

Previously the best method for computing matrix complex polynomials was QEP \cite{LowSu2024}, but its initial state preparation cost is not optimal. In \cite{LowSu2024quantumlinearalgorithmoptimal}, the initial state preparation cost was optimized. Here, we restrict both methods to real polynomials for comparison. 

In the comparison, we follow the setting of \cite{LowSu2024quantumlinearalgorithmoptimal} and assume that we need to achieve $f(A/\alpha)$. Because $\rho(A/\alpha)\leq1$, we can let $\Gamma$ be the boundary of $\{z:|z|\leq1+a\}$, and $a=\mathcal{O}(1)$. Then we have the comparison given in Table \ref{compare:Matrix_Polynomial}. In our work, $B$ is the maximum value of $f(z)$ on $\Gamma$, and in \cite{LowSu2024quantumlinearalgorithmoptimal}, it is assumed that $A$ has real spectra, and $\hat{B}$ is the maximum value of $f(z)$ on $[-1,1]$. Since $\Gamma$ encompasses $[-1,1]$, by the maximum modulus principle, $\hat{B}\leq B$. 
\REV{However, the contour integral approach does not have explicit dependence on the degree of the polynomial, indicating potential speedup when the degree of the polynomial is high and the growth is not rapid. 
Nevertheless, such a degree dependence might re-enter the complexity implicitly through other parameters like $\hat{B}$, so whether or not the quantum contour integral approach can be genuinely advantageous must be analyzed on a case-by-case basis. 
Additionally, as displayed in Table \ref{compare:Matrix_Polynomial}, the quantum contour integral approach improves the $\kappa_S$ dependence from quadratic to poly-logarithmic scalings for circular contour. 
}

\begin{table}[ht]
\centering
\setlength{\tabcolsep}{3pt}
\renewcommand{\arraystretch}{1.8}
\scalebox{0.9}{
\begin{tabular}{c|c}
\hline\hline
\textbf{Methods} & \textbf{Query complexities}  \\
\hline
\REV{QEP \cite[Theorem 7]{LowSu2024quantumlinearalgorithmoptimal}} & \REV{$\widetilde{\mathcal{O}} \left(\frac{\hat{B} \kappa_S^2 n }{\left\|f(A/\alpha) \ket{\psi}\right\|}  \log \left(\frac{1}{\epsilon}\right) \right) $}  \\
\hline
This work & \REV{$\widetilde{\mathcal{O}}\left(\frac{B \log^2\kappa_S}{\left\|f(A/\alpha)\ket{\psi}\right\|}\mathrm{log}\left(\frac{1}{\epsilon}\right)  \right)$ } \\
\hline\hline
\end{tabular}
}
\caption{Comparison between contour integral and the previous approach for implementing polynomial $f(A)$. Our task is given in Definition \ref{definition:main_problem}, and all methods use amplitude amplification to amplify the probability. We assume that $A = SDS^{-1}$ is a diagonalizable matrix, $\kappa_S = \|S\|\|S^{-1}\|$ is the condition number, $\alpha$ is the block-encoding factor of A and $n$ is the degree of $f$. In our work, $B$ is the maximum value of $f(z)$ on $\Gamma$, and in \cite{LowSu2024quantumlinearalgorithmoptimal}, it is assumed that $A$ has real spectra, and $\hat{B}$ is the maximum value of $f(z)$ on $[-1,1]$.}
\label{compare:Matrix_Polynomial}
\end{table}
\REV{
Here we discuss a particularly favorable case for the contour integral approach when we implement a matrix Schur function and the spectrum of $A$ is strictly inside the unit circle. 
Specifically, assume that $\rho(A)<1$, $\|A\| \leq 1$ and that $A=SDS^{-1}$ is diagonalizable. 
Taking $\Gamma=\{z:|z|=1\}$ gives $l=2\pi$. 
Let $f$ is a Schur function, which is defined to be holomorphic in a neighborhood of the closed unit disk and $\max_{|z|=1}|f(z)| = 1$. 
Therefore, the query complexity of our method is bounded by
\begin{equation*}
    \mathcal{O}\left(
    \frac{(\alpha+1)\log^2 \kappa_S}
    {\|f(A)\ket{\psi}\| (1-\rho(A))^2 \log^2(1/\rho(A)) }
    \log\left(
    \frac{\log\kappa_S}
    {\epsilon \|f(A)\ket{\psi}\|(1-\rho(A)) \log(1/\rho(A))}
    \right)
    \right). 
\end{equation*}
In particular, for the example of a Schur polynomial
\begin{equation*}
    f_d(z)=\frac{1+z^d}{2},
\end{equation*}
we have $\max_{|z|=1}|f_d(z)| =1$ for every $d$, while $f_d$ has degree $d$ and is not a monomial. 
When $\rho(A)=r<1$ is a fixed constant, $A^d$ decays spectrally and $\|f_d(A)\ket{\psi}\|$ remains bounded from below for large $d$ in typical cases. 
For example, if $d$ is sufficiently large such that $\kappa_S r^d<1/2$, then
\begin{equation*}
    \|f_d(A)\ket{\psi}\|
    \ge
    \frac{1-\|A^d\|}{2}
    \ge
    \frac{1-\kappa_S r^d}{2} \geq \frac{1}{4}.
\end{equation*}
Thus the query complexity of the contour integral method becomes $\mathcal{O}\left( (\alpha+1)(\log(1/\epsilon)+\log\log \kappa_S)\log^2\kappa_S  \right)$, which has no polynomial dependence on the degree $d$. 
By contrast, QEP implements a degree-$d$ polynomial through a construction whose cost explicitly depends on $d$. 
Consequently, for bounded high-degree Schur polynomials and matrices with spectrum uniformly separated from the unit circle, the contour integral approach gives a strict asymptotic improvement over QEP in the degree dependence. }

If we use the ancilla-efficient algorithm, suppose that we need to estimate the observable value of $O$ under the state $f(A)\ket{\psi}$, then according to Theorem \ref{theorem:Ancilla-efficient_complexity}, in each run of the quantum circuit we need
\REV{$$
\mathcal{O} \left(\log \kappa_S~\mathrm{log}\left(\frac{B\kappa_S\|O\|\|f(A)\|}{\epsilon}\right)\right)
$$}
queries to $U_A$, $1$ queries to $U_{\psi}$ and $\tilde{a}+3$ ancilla qubits, and we need to repeat the circuit $K$ times, where
\REV{$$
K =\mathcal{O} \left(\frac{\|O\|^2 \ln (1 / \xi_2)(B \log \kappa_S)^4}{\epsilon^2 }\right).
$$}

\subsection{Other applications}
\label{sec:other_applications}

\REV{Besides the applications discussed above, the contour integral approach can also be applied to several other matrix functions. 
We briefly discuss these examples here to illustrate the flexibility of the method. 
The common feature of these applications is that the target transformation is naturally an eigenvalue transformation rather than a singular value transformation, especially when the input matrix is non-normal. 
Therefore, these examples are not directly covered by standard QSVT applied to a block-encoding of $A$.}

\REV{A basic example is Hamiltonian simulation. 
For a Hermitian Hamiltonian $H$, the target operation is $\mathrm{e}^{-\mathrm{i}HT}$, which can be written as the matrix function $f(A)=\mathrm{e}^{TA}$ with $A=-\mathrm{i}H$. 
Since the spectrum of $A$ lies on the imaginary axis, one can choose a thin contour enclosing this imaginary interval and then apply the contour integral algorithm. 
We give the detailed calculation in Appendix \ref{app:Hamiltonian_simulation}. 
As discussed there, this example mainly serves as a prototype application of the contour integral framework. 
For Hermitian Hamiltonian simulation alone, QSVT-based methods are more efficient. }

\REV{Another important example is the matrix sign function recently discussed in \cite{WangLiu2026}. 
Assume that the spectrum of $A$ does not intersect the imaginary axis, namely
\begin{equation}
    \operatorname{spec}(A)\cap \mathrm{i}\mathbb{R}=\emptyset .
\end{equation}
Then the spectrum can be separated into the part in the open right half-plane and the part in the open left half-plane. 
Let $\Gamma_+$ be a contour enclosing the eigenvalues with positive real parts, and let $\Gamma_-$ be a contour enclosing the eigenvalues with negative real parts. 
Define
\begin{equation}
    f(z)
    =
    \begin{cases}
        1, & \operatorname{Re}(z)>0,\\
        -1, & \operatorname{Re}(z)<0.
    \end{cases}
\end{equation}
Equivalently, one can regard $f$ as a function that is holomorphic separately in two disjoint domains containing the two spectral components. 
Then
\begin{equation}
    \operatorname{sign}(A)
    =
    \frac{1}{2\pi \mathrm{i}}
    \oint_{\Gamma_+}
    (zI-A)^{-1}\mathrm dz
    -
    \frac{1}{2\pi \mathrm{i}}
    \oint_{\Gamma_-}
    (zI-A)^{-1}\mathrm dz .
\end{equation}
Thus, by applying the contour integral algorithm on the two contours and combining the two contributions by LCU, one obtains an implementation of the matrix sign function. 
Since the matrix sign function is an important tool in numerical linear algebra and control theory, with applications to computing invariant subspaces, model reduction, and solving Sylvester, Lyapunov, and algebraic Riccati equations~\cite{Higham2008,WangLiu2026}, the contour integral approach can readily be used in these applications as well.  
}

\REV{The same idea also applies to Riesz spectral projectors. 
If $\Gamma$ encloses only a selected subset of the spectrum of $A$, then
\begin{equation}
    P_{\Gamma}
    =
    \frac{1}{2\pi \mathrm{i}}
    \oint_{\Gamma}
    (zI-A)^{-1}\mathrm dz
\end{equation}
is the spectral projector onto the corresponding invariant subspace. 
Such projectors are useful for spectral filtering, invariant subspace extraction, stability analysis, and non-Hermitian eigenvalue problems. }

\REV{Other holomorphic matrix functions can be treated similarly whenever a suitable branch is chosen on a domain containing the spectrum. 
Examples include the matrix logarithm, fractional powers $A^\nu$ with $0<\nu<1$, inverse fractional powers, and resolvent-type filters. 
These functions appear in numerical linear algebra, fractional differential equations, model reduction, and stability analysis. 
The contour integral framework provides a uniform way to implement them, as long as the contour avoids singularities and branch cuts of the scalar function. 
These examples demonstrate that the contour integral approach is not restricted to the specific functions analyzed in detail in this section. 
Its main strength is that it implements holomorphic eigenvalue transformations for general matrices. 
}

\section*{Acknowledgments}

DA and SJ acknowledge funding from Quantum Science and Technology - National Science and Technology Major Project via Project 2024ZD0301900, and the support by The Fundamental Research Funds for the Central Universities, Peking University.

{\small
\bibliographystyle{quantum}
\bibliography{reference}
}
\appendix

\section{Upper bound of~\texorpdfstring{$\left\|(z I_N -A)^{-1}\right\| $}{}}
The necessary and sufficient condition for the integral formula in Eq. \eqref{function:f(A)} to be valid is that $\Gamma$ covers all eigenvalues of $A$. But at this point, it is difficult to estimate the lower bound of the singular value of $z I_N -A$. We have that $\sigma_{min}(zI_N-A)=1/\|(z_k I_N -A)^{-1}\|$. In our analysis, we assume that $\|z I_N -A\|\leq \gamma$. If $A$ can be diagonalized, let $A = SDS^{-1}$ and $a$ is the lower bound of the distance from $\Gamma$ to the eigenvalues of $A$ then:

\begin{equation}
    \left\|(z I -A)^{-1}\right\| = \left\|S(z I -D)^{-1}S^{-1}\right\| \leq  \left\|S\right\|\left\|S^{-1}\right\|\left\|(z I -D)^{-1}\right\| \leq \frac{\kappa_S}{a}.
\end{equation}

However, for the general case, the singular value bound of the matrix $z_k I_N -A$ can be very bad. For example, let 
\begin{equation}
    L_n = \begin{bmatrix}
0 & 1 & 0 & \cdots & 0 \\
0 & 0 & 1 & \cdots & 0 \\
\vdots & \vdots & \ddots & \ddots & \vdots \\
0 & 0 & \cdots & 0 & 1 \\
0 & 0 & \cdots & 0 & 0
\end{bmatrix}.
\end{equation}
For $a>0$, $A=aL_N$ has spectral radius $0$, and $A^N=0$. Therefore, we have
\begin{equation}
    (zI_N-A)^{-1} = z^{-1}\left(I_N-\frac{a}{z}L_N\right)^{-1}=z^{-1}\sum_{k=0}^{N-1}\left(\frac{a}{z}L_N\right)^k= \begin{bmatrix}
z^{-1} & a z^{-2} & a^2 z^{-3} & \cdots & a^{N-1} z^{-N} \\
0 & z^{-1} & a z^{-2} & \cdots & a^{N-2} z^{-(N-1)} \\
\vdots & \vdots & \ddots & \ddots & \vdots \\
0 & 0 & \cdots & z^{-1} & a z^{-2} \\
0 & 0 & \cdots & 0 & z^{-1}
\end{bmatrix}, 
\end{equation}
because for $k=0,1,\cdots,N-1$, $\|L^k\|=1$, we have $$\left\|(z I -A)^{-1}\right\|\leq \|z\|^{-1}\sum_{k=0}^{N-1}\left(\frac{a}{\|z\|}\right)^k=\frac{1-\left(a/\|z\|\right)^{N}}{\|z\|-a},$$ which exponentially depends on $N$ when $a/\|z\|>1$.

In fact, the exponential dependence on $N$ is also the worst case. Assume $A$ has Jordan standard form $A=SJS^{-1}$, then
\begin{equation}
    \left\|(z I -A)^{-1}\right\| = \left\|S(z I -J)^{-1}S^{-1}\right\| \leq  \left\|S\right\|\left\|S^{-1}\right\|\left\|(z I -J)^{-1}\right\| \leq \left\|(z I -J)^{-1}\right\|\kappa_S.
\end{equation}
Suppose that $z I -J$ is constructed by $n$ Jordan blocks, which means
\begin{equation}
    z I -J = \begin{bmatrix}
J_1 & 0 & \cdots & 0 \\
0 & J_2 &  \cdots & 0 \\
\vdots & \vdots & \ddots  & \vdots \\
0 & 0 & \cdots &  J_n
\end{bmatrix}. 
\end{equation}
For every Jordan block $J_k$, we have
\begin{equation}
    J_k^{-1} = (\lambda_kI_{N_k}+L_{N_k})^{-1} =\lambda_k^{-1}\sum_{k=0}^{N-1}(-\lambda_k^{-1}L_{N_k})^k,
\end{equation}
and $\|J_k\| \leq \frac{1-\|\lambda\|^{-N_k}}{\|\lambda\|-1}\leq \frac{1-\|\lambda\|^{-N}}{\|\lambda\|-1}$. Therefore, we have  
\begin{equation}
    \left\|(z I -A)^{-1}\right\| = \left\|S(z I -J)^{-1}S^{-1}\right\| \leq  \left\|S\right\|\left\|S^{-1}\right\|\left\|(z I -J)^{-1}\right\| \leq \frac{1-a^{-N}}{a-1}\kappa_S.
\end{equation}
Where $a$ is the minimum distance from curve $\Gamma$ to the eigenvalues of $A$.

\section{Dissipative conditions}\label{app:dissi_cond}
In this part, we show that $A + A^\dagger \preceq 0$ is stronger than $\mathrm{Re}(\lambda_A) \leq 0$. In fact, we have the following lemma.
\label{lemma:Dissipative_conditions}
\begin{lemma}
If $A + A^\dagger$ is negative semi-definite ($A + A^\dagger \preceq 0$), then all eigenvalues of $A$ have non-positive real parts ($\operatorname{Re}(\lambda_A) \leq 0$). However, the converse is not true.
\end{lemma}

\begin{proof}
First, we show that if $A + A^\dagger \preceq 0$, then $\mathrm{Re}(\lambda_A) \leq 0$. Let $\lambda$ be an eigenvalue of $A$, and $v\neq0$ is a corresponding eigenvector. Then $Av=\lambda v$. Due to that $A + A^\dagger \preceq 0$, we have
\begin{equation}
    v^\dagger(A + A^\dagger)v\leq0.
\end{equation}
Therefore,
\begin{equation}
    v^\dagger Av+v^\dagger A^\dagger v=(\lambda+\bar{\lambda})v^\dagger v=2\mathrm{Re}(\lambda)\|v\|^2\leq0.
\end{equation}
Since $v\neq0$, we have $\|v\|^2>0$ and $\mathrm{Re}(\lambda)\leq 0$. Due to the arbitrariness of $\lambda$, all eigenvalues of $A$ have non-positive real parts. 

Now we show that the converse fails. Consider the matrix $A = \begin{bmatrix} 0 & 1 \\ 0 & 0 \end{bmatrix}$.
\begin{equation}
    \det(\lambda I - A) = \det\begin{bmatrix} \lambda & -1 \\ 0 & \lambda \end{bmatrix} = \lambda^2.
\end{equation}
Then, $\lambda_A=0,0$ and all eigenvalues have real part $0 \leq 0$. However,
\begin{equation}
    A + A^\dagger = \begin{bmatrix} 0 & 1 \\ 0 & 0 \end{bmatrix} + \begin{bmatrix} 0 & 0 \\ 1 & 0 \end{bmatrix} = \begin{bmatrix} 0 & 1 \\ 1 & 0 \end{bmatrix}.
\end{equation}
Let $v=[1,1]^\dagger$, then
\begin{equation}
    v^\dagger (A + A^\dagger) v = \begin{bmatrix} 1 & 1 \end{bmatrix} \begin{bmatrix} 0 & 1 \\ 1 & 0 \end{bmatrix} \begin{bmatrix} 1 \\ 1 \end{bmatrix} =  2 > 0.
\end{equation}
Therefore, $A + A^\dagger$ is not negative semi-definite.

\end{proof}

\section{Complexity of single-ancilla LCU with block-encoding}
\label{section:Ancilla-efficient complexity}
In this section we prove the complexity of single-ancilla LCU. The proof steps are only slight modifications from the proof of \cite[Theorem 8]{Chakraborty2024} to fit the case of block-encoding. However, for completeness, we present a self-contained proof. First, we have the following theorem.

\begin{theorem}
    Suppose $P$ and $Q$ are operators such that $\|P-Q\| \leq \xi$ for some $\xi \in[0,1]$, and $U_Q$ is an $(\alpha,a,0)$-block-encoding of $Q$. Furthermore, let $\ket{\psi}$ be any state and $\ket{\widetilde{\psi}}=\ket{0^a}\otimes\ket{\psi}$, $O$ be some Hermitian operator with spectral norm $\|O\|$ and $\widetilde{O}=\ket{0^a}\bra{0^a}\otimes O$. Then, if $\|P\| \geq 1$, the following holds:
\begin{equation}
    \left\|\bra{\psi}P^\dagger O P\ket{\psi}-\alpha^2\bra{\widetilde{\psi}}U_Q^\dagger \widetilde{O} U_Q\ket{\widetilde{\psi}}\right\| \leq 3\|O\|\|P\| \xi.
\end{equation}
\label{theorem:accuracy_of_observation}
\end{theorem} 
\begin{proof}
Let $\rho=\ket{\psi}\bra{\psi}$, then 
\begin{equation}
    \bra{\psi}P^\dagger O P\ket{\psi}=\operatorname{Tr}[O P\ket{\psi}\bra{\psi}P^\dagger]=\operatorname{Tr}[O P\rho P^\dagger],
\end{equation}
\begin{equation}
   \begin{aligned}
\widetilde{O}U_Q\ket{\widetilde{\psi}}=\left(\ket{0^{a}}\bra{0^{a}}\otimes O\right)\left(\ket{0^{a}}\otimes \frac{Q}{\alpha}\ket{\psi}+\ket{\perp}\right)= \frac{1}{\alpha}\left(\ket{0^{a}}\otimes O Q\ket{\psi}\right),
\end{aligned}
\end{equation}

\begin{equation}
\label{equation:OQrhoQdagger}
\begin{aligned}
\bra{\widetilde{\psi}}U_Q^\dagger\widetilde{O}U_Q\ket{\widetilde{\psi}}=\frac{1}{\alpha^2}\left(\bra{0^{a}}\otimes\bra{\psi} Q^\dagger+\bra{\perp}\right)\left(\ket{0^{a}}\otimes O Q\ket{\psi}\right)=\frac{1}{\alpha^2}\bra{\psi}Q^\dagger O Q\ket{\psi}=\frac{1}{\alpha^2}\operatorname{Tr}[O Q\rho Q^\dagger].
\end{aligned}
\end{equation}
Using the identity
\begin{equation}
    P \rho P^{\dagger}-Q \rho Q^{\dagger}=(Q-P) \rho\left(P^{\dagger}-Q^{\dagger}\right)+P \rho\left(P^{\dagger}-Q^{\dagger}\right)+(P-Q) \rho P^{\dagger},
\end{equation}
we have
\begin{equation}
\begin{aligned}
\left\|P \rho P^{\dagger}-Q \rho Q^{\dagger}\right\|_1 & \leq\|Q-P\|\left\|\rho\left(P^{\dagger}-Q^{\dagger}\right)\right\|_1+\|P\|\left\|\rho\left(P^{\dagger}-Q^{\dagger}\right)\right\|_1+\|P-Q\|\left\|\rho P^{\dagger}\right\|_1 \\
& \leq\|P-Q\|^2\|\rho\|_1+\|P\|\|P-Q\|\|\rho\|_1+\|P-Q\|\|P\|\|\rho\|_1 \\
& \leq\|P-Q\|^2+2\|P-Q\|\|P\|.
\end{aligned}
\end{equation}
Then, together with the Holder's inequality with $p=\infty$ and $q=1$
\begin{equation}
\label{equation:trdistance}
\left\|\operatorname{Tr}\left[OP \rho P^{\dagger}
\right]-\operatorname{Tr}\left[O Q \rho Q^{\dagger}\right]\right\| \leq\|O\| \cdot\left\|P \rho P^{\dagger}-Q \rho Q^{\dagger}\right\|_1,
\end{equation}
we obtain
\begin{equation}
\begin{aligned}
\left\|\operatorname{Tr}\left[O P \rho P^{\dagger}\right]-\operatorname{Tr}\left[O Q \rho Q^{\dagger}\right]\right\| & \leq\|O\|\|P-Q\|^2+2\|O\|\|P\|\|P-Q\| \\
& \leq \xi^2\|O\|+2\|O\|\|P\| \xi \\
& \leq 3 \xi\|O\|\|P\|. 
\end{aligned}
\end{equation}
Together with \cref{equation:OQrhoQdagger} we get 
\begin{equation}
    \left\|\bra{\psi}P^\dagger O P\ket{\psi}-\alpha^2\bra{\widetilde{\psi}}U_Q^\dagger \widetilde{O} U_Q\ket{\widetilde{\psi}}\right\| \leq 3\|O\|\|P\| \xi.
\end{equation}
\end{proof}

With this result, we can prove the Theorem \ref{theorem:Ancilla-efficient_convergence} as below.

\begin{theorem}[Estimating expectation values of observables]
\label{theorem:Ancilla-efficient_complexity_proof}
Let $\epsilon, \xi_1, \xi_2 \in(0,1)$ be some parameters. $U_k$ is an $(1,a,0)$-block-encoding of $p^\diamond(\frac{z_k I-A}{\alpha+|z_k|})^\dagger$. Let $\ket{\psi}$ be some initial state and $\ket{\widetilde{\psi}} =\ket{0^a}\otimes\ket{\psi}$, $O$ be some observable and $\widetilde{O}=\ket{0^a}\bra{0^a}\otimes O$.
Suppose that $\left\|f(A)-\sum_j c_j p^\diamond\left(\frac{z_k I-A}{\alpha+|z_k|}\right)^\dagger\right\| \leq \xi_1$, where 
$$
\xi_1 \leq \frac{\epsilon}{6\|O\|\|f(A)\|}.
$$
Furthermore, let
$$
T \geq \frac{8\|O\|^2 \ln ( 2/ \xi_2)\|c\|_1^4}{\epsilon^2}.
$$
Then, Algorithm \ref{algorithm:random_LCU} estimates $\mu$ such that
$$
\left|\mu-\bra{\psi}f(A)^\dagger O f(A)\ket{\psi}\right| \leq \epsilon,
$$
with probability at least $(1-\xi_2)^2$, using of one ancilla qubit and $T$ repetitions of the quantum circuit in Algorithm \ref{algorithm:random_LCU}, Step 3.
\end{theorem}

\begin{proof}

First define $\rho_0=\ket{\psi}\bra{\psi}$, $\widetilde{\rho_0}=\ket{\widetilde{\psi}}\bra{\widetilde{\psi}}$, and let $$g(A)=\sum_j c_j p^\diamond\left(\frac{z_k I-A}{\alpha+|z_k|}\right)^\dagger.$$

Following Algorithm \ref{algorithm:random_LCU}, the initial state $\rho_1=\ket{+}\bra{+}\otimes\widetilde{\rho_0}$ is mapped to

\begin{equation}
\begin{aligned}
\rho^{\prime} & =\tilde{V}_2 \tilde{V}_1 \rho_1 \tilde{V}_1^{\dagger} \tilde{V}_2^{\dagger} \\
& =\frac{1}{2}\left[|0\rangle\langle 0| \otimes V_2 \widetilde{\rho_0} V_2^{\dagger}+|0\rangle\langle 1| \otimes V_2 \widetilde{\rho_0} V_1^{\dagger}+|1\rangle\langle 0| \otimes V_1 \widetilde{\rho_0} V_2^{\dagger}+|1\rangle\langle 1| \otimes V_1 \widetilde{\rho_0} V_1^{\dagger}\right].
\end{aligned}
\end{equation}
Measuring the observable $X\otimes\widetilde{O}$ yields the expectation value
\begin{equation}
\operatorname{Tr}\left[(X \otimes \widetilde{O}) \rho^{\prime}\right]=\frac{1}{2} \operatorname{Tr}\left[\widetilde{O}\left(V_1 \widetilde{\rho_0} V_2^{\dagger}+V_2 \widetilde{\rho_0} V_1^{\dagger}\right)\right]. 
\end{equation}
Taking expectation over the outcome of the $j^{\text {th }}$ iteration gives
\begin{equation}
\mathbb{E}\left[\mu_j\right]=\mathbb{E}\left[\operatorname{Tr}\left[(X \otimes \widetilde{O}) \rho^{\prime}\right]\right]=\frac{1}{\|c\|_1^2} \operatorname{Tr}\left[\widetilde{O} \left(\sum_j c_j U_j\right) \widetilde{\rho_0} \left(\sum_j c_j U_j\right)^{\dagger}\right]=\frac{1}{\|c\|_1^2} \operatorname{Tr}\left[O g(A) \rho_0 g(A)^{\dagger}\right].
\end{equation}
The Hoeffding's inequality implies that, since each random variable lies in the range $-\|O\|\|c\|_1^2 \leq\|c\|_1^2 \mu_j \leq+\|O\|\|c\|_1^2$, 
\begin{equation}
\operatorname{Pr}\left[\left|\mu-\operatorname{Tr}\left[O g(A) \rho_0 g(A)^{\dagger}\right]\right| \geq \epsilon / 2\right] \leq 2 \exp \left[-\frac{T \epsilon^2}{8\|c\|_1^4\|O\|^2}\right]. 
\end{equation}
This shows that for
\begin{equation}
    T \geq \frac{8\|O\|^2 \ln (2 / \xi_2)\|c\|_1^4}{\epsilon^2}, 
\end{equation}
\begin{equation}
\label{equation:mu-ogrhog}
    \left|\mu-\operatorname{Tr}\left[O g(A) \rho_0 g(A)^{\dagger}\right]\right| \leq \epsilon / 2, 
\end{equation}
with probability at least $1-\xi_2$. Note that $\|f(A)-g(A)\| \leq \xi_1$. For any such operators that are at most $\xi_1$-separated, we can use Theorem \ref{theorem:accuracy_of_observation} to obtain
\begin{equation}
\label{equation:ofrhof-ogrhog}
\left|\operatorname{Tr}\left[O f(A) \rho_0 f(A)^{\dagger}\right]-\operatorname{Tr}\left[O g(A) \rho_0 g(A)^{\dagger}\right]\right| \leq 3\|O\|\|f(A)\| \xi_1 \leq \epsilon / 2,
\end{equation}
for $\xi_1$ upper bounded as in the statement of this theorem. So, combining \cref{equation:mu-ogrhog} and \cref{equation:ofrhof-ogrhog} we have
\begin{equation}
    \left|\mu-\operatorname{Tr}\left[O f(A) \rho_0 f(A)^{\dagger}\right]\right| \leq \epsilon,
\end{equation}
which completes the proof.
\end{proof}

\section{Application in Hamiltonian simulation}
\label{app:Hamiltonian_simulation}
Consider the problem of the time-independent Hamiltonian simulation. For a Hermitian matrix $H$, we hope to solve the linear differential equation
\begin{equation}
\label{equation:Hamiltonian_simulation}
    \mathrm{i}\frac{\mathrm{d}}{\mathrm{d}t}x(t)=H x(t), \quad x(0)=\ket{\psi},
\end{equation}
and the solution can be written as
\begin{equation}
    x(T) = \mathrm{e}^{-\mathrm{i} HT}\ket{\psi}.
\end{equation}
Therefore, after implementing function $\mathrm{e}^{-\mathrm{i} HT}$ through contour integral and applying $\mathrm{e}^{-\mathrm{i} HT}$ to the initial state, we get a method to solve time-independent Hamiltonian simulation. Let $f(x) = e^{Tx}$ with $T>0$, because of that Hermitian matrix $H$ only has real eigenvalues, the eigenvalues of $-\mathrm{i}H$ are on the imaginary axis. Since when $\|H\|>1$, we can always absorb $\|H\|$ into $T$, we assume $\|H\|\leq1$ and the eigenvalues of $-\mathrm{i}H$ are on the line between $- \mathrm{i}$ and $ \mathrm{i}$ on the complex plane, and choose $\Gamma$ to be an area that covers the eigenvalues (Fig.\ref{fig:spectrum_contour_2}).

\begin{figure}[ht]
    \centering
    \begin{tikzpicture}[>=Stealth, scale=1.2] 

        \draw[->, thick] (-1.5,0) -- (1.5,0) node[right] {$\mathrm{Re}$};
        \draw[->, thick] (0,-1.5) -- (0,1.5) node[above] {$\mathrm{Im}$};;

        \draw[blue, thick] (0.1,1.1) arc (85:95:1.1);
        \draw[blue, thick] (-0.1,1.1) -- (-0.1,-1.1);
        \draw[blue, thick] (-0.1,-1.1) arc (265:275:1.1);
        \draw[blue, thick] (0.1,1.1) -- (0.1,-1.1);
        \draw[red, ultra thick] (0,-1) -- (0,1);

        \node[left] at (-0.02, 1)  {$\mathrm{i}$};
        \node[left] at (-0.02, -1) {$-\mathrm{i}$};
        \node[right] at (0.3, 0.8) {$f(z)=\mathrm{e}^{Tz}$};

    \end{tikzpicture}
    \caption{In time-independent Hamiltonian simulation, the spectrum of $-\mathrm{i} H$ and the covering contour.}
    \label{fig:spectrum_contour_2}
\end{figure}

We choose $\Gamma$ as the boundary of $\{z:|z|\leq 1+a, -a \leq \mathrm{Re}(z)\leq a\}$, for $a$ is a parameter to be determined. The maximum value of $|f(z)|$ on $\Gamma$ is $B = e^{aT}$, \REV{and the upper bound on the tangent variation rate of $\Gamma$ is $K_{\Gamma}=\frac{1}{1+a}$.} Since the Hermitian matrix can be unitarily diagonalized, let $-\mathrm{i}H = UDU^\dagger$, then
\begin{equation}
    \left\|(z I +\mathrm{i}H)^{-1}\right\| = \left\|U(z I -D)^{-1}U^\dagger\right\| = \left\|(z I -D)^{-1}\right\| \leq \frac{1}{a}.
\end{equation}
Therefore, we can let $\gamma=\frac{1}{a}$, and $\Gamma$ is on the disk $\{|z| \leq R = \pi + a\}$. According to Theorem \ref{theorem:general_complexity}, we need $$\mathcal{O}\left(\frac{\mathrm{e}^{aT}(1 + a)(\alpha+1 + a)}{a^2\left\|f(-\mathrm{i}H)\ket{\psi}\right\|}\mathrm{log}\left(\frac{\mathrm{e}^{aT}(1 + a)}{\|f(-\mathbf{i}H)\ket{\psi}\|a\epsilon} \right) \right)$$ queries to $U_H$, $$\mathcal{O}\left(\frac{\mathrm{e}^{aT}(1 + a)^2}{a\left\|f(-\mathrm{i}H) \ket{\psi}\right\|}\right) $$ queries to $U_{\psi}$ and $$\mathcal{O}\left(\mathrm{log}\left(\frac{1}{a\|f(-\mathrm{i}H) \ket{\psi}\|\epsilon}\right) \right)+\tilde{a}+2$$ ancilla qubits, while $\tilde{a}$ ancilla qubits are used to construct $U_A$. Note that $f(-\mathrm{i}H) \ket{\psi}$ is unitary and $\left\|f(-\mathrm{i}H) \ket{\psi}\right\|=1$. Due to that $\|H\|\leq1$, we assume $\alpha=\mathcal{O}(1)$. We would choose $a = \mathrm{min}\left\{\frac{1}{T},1\right\}$, then $\alpha+1 + a=\mathcal{O}(1)$, and we need $\mathcal{O}\left(T^2\mathrm{log}\left(\frac{T}{\epsilon}\right)  \right)$  queries to $U_H$, $\mathcal{O}\left(T\right)$  queries to $U_{\psi}$ and $\mathcal{O}\left(\mathrm{log}\left(\frac{T}{\epsilon}\right)\right)+\tilde{a}+2$ ancilla qubits.

\begin{table}[ht]
\centering
\setlength{\tabcolsep}{15pt}
\renewcommand{\arraystretch}{1.8}
\begin{tabular}{c|c}
\hline\hline
\textbf{Methods} & \textbf{Query complexities}  \\
\hline
QSVT\cite{GilyenSuLowEtAl2019} & $\mathcal{O}\left(T +\frac{\mathrm{log}(1/\epsilon)}{\mathrm{log}(\mathrm{e}+\mathrm{log}(1/\epsilon)/T)}\right) $ \\
\hline
This work & $\mathcal{O}\left(T^2\mathrm{log}\left(\frac{T}{\epsilon}\right)  \right)$  \\
\hline\hline
\end{tabular}
\caption{Comparison between contour integral and the previous approach for solving Eq. \eqref{equation:Hamiltonian_simulation}. Our task given on Definition \ref{definition:main_problem}, and both methods use amplitude amplification to amplify the probability. $T$ is the evolution time, and $\epsilon$ is the tolerated error in the output state.}
\label{compare:Hamiltonian_simulation}
\end{table}

A comparison is given in Table \ref{compare:Hamiltonian_simulation}. If limited to calculating Hamiltonian simulations, using QSVT\cite{GilyenSuLowEtAl2019} directly is a better approach. The advantage of the contour integral method is in calculating general matrix functions, which we will see later. 

If we use the ancilla-efficient algorithm, suppose that we would observe $\mathrm{e}^{-\mathrm{i} HT}\ket{\psi}$ under $O$, then according to Theorem \ref{theorem:Ancilla-efficient_complexity}, each time we need
$$
\mathcal{O} \left(T\mathrm{log}\left(\frac{T\|O\|}{\epsilon}\right)\right)
$$ queries to $U_A$, $1$ queries to $U_{\psi}$ and $\tilde{a}+3$ ancilla qubits, and we need to repeat the circuit $K$ times, where
$$
K =\mathcal{O} \left(\frac{\|O\|^2 \ln (1 / \xi_2)T^4}{\epsilon^2}\right).
$$

\end{document}